\documentclass[11pt,leqno,textwidth=10cm]{amsart}
\usepackage{amssymb}
\allowdisplaybreaks
\usepackage{mathrsfs}
\usepackage{mathtools}
\usepackage{abstract}
\usepackage[hidelinks]{hyperref}
\usepackage{marginnote}

\usepackage{etoolbox}
\makeatletter
\patchcmd{\@maketitle}{\newpage}{}{}{} 
\makeatother
\numberwithin{equation}{section}

\newtheorem{thm}{Theorem}
\newtheorem{lem}{Lemma}
\newtheorem{prop}{Proposition}
\newtheorem{corol}{Corollary}
\newtheorem{defn}{Definition}
\newtheorem{rem}{Remark}

\newcommand{\p}{\partial}
\newcommand{\hdel}{\hat{\partial}_0}
\newcommand{\oddk}{{\mathring{k}}}
\newcommand{\ck}{{\left \lceil{k /2}\right \rceil }}
\newcommand{\fk}{{\left \lfloor{k/2}\right \rfloor}}
\newcommand{\Eg}{{E^g}}
\newcommand{\tphi}{\tilde \phi}
\newcommand{\hN}{\widehat{N}}
\newcommand{\Ephi}[1]{{\mathcal{E}_{#1}(\phi)}}
\newcommand{\tEphi}[1]{\tilde{ \mathcal{E}}_{#1}(\phi)}

\newcommand{\mcr}[1]{\mathscr{#1}}
\newcommand{\eq}[1]{\begin{equation} #1 \end{equation}}

\begin{document}

%%%%%%%%%%%%%%%%%%%%%%%%%%%%%%%%%%%%%%%%%%%%%%%%%%%%%%%%%%%%%%%%%%%%%%%
%%%%%%%%%%%%%%%%%%%%%%%%%%%%%%%%%%%%%%%%%%%%%%%%%%%%%%%%%%%%%%%%%%%%%%%

\title[Attractors of the Einstein-Klein Gordon system]{Attractors of the Einstein-Klein-Gordon system}
\author[D.~Fajman, Z.~Wyatt]{David Fajman, Zoe Wyatt}

\date{\today}

\subjclass[2010]{35Q75; 83C05;  35B35}
\keywords{Nonvacuum Einstein flow, Einstein-Klein-Gordon system, Nonlinear Stability, Milne model}

\address{
\begin{tabular}[h]{l@{\extracolsep{8em}}l} 
David Fajman & Zoe Wyatt \\
Faculty of Physics & Maxwell Institute for Mathematical Sciences \\ 
University of Vienna & School of Mathematics \\
Boltzmanngasse 5 & University of Edinburgh  \\
1090 Vienna, Austria & Edinburgh, EH9 3FD, UK \\
David.Fajman@univie.ac.at & zoe.wyatt@ed.ac.uk       \\
&\\
&\qquad \qquad AND\\
&\\
&Faculty of Physics \\
&University of Vienna \\
&Boltzmanngasse 5 \\
&1090 Vienna, Austria \\
\end{tabular}
}

%\address{
%David Fajman\newline
%Faculty of Physics, 
%University of Vienna\\ \newline
%Boltzmanngasse 5\\ \newline
%1090 Vienna, Austria\\ \newline
%David.Fajman@univie.ac.at
%}
%
%
%\address{
%Zoe Wyatt\newline
%Maxwell Institute for Mathematical Sciences\newline
%School of Mathematics\newline
%University of Edinburgh\newline
%Edinburgh, EH9 3FD\newline
%UK\newline
%zoe.wyatt@ed.ac.uk
%}

\maketitle
%%%%%%%%%%%%%%%%%%%%%%%%%%%%%%%%%%%%%%%%%%%%%%%%%%%%%%%%%%%%%%%%%%%%%%%
%%%%%%%%%%%%%%%%%%%%%%%%%%%%%%%%%%%%%%%%%%%%%%%%%%%%%%%%%%%%%%%%%%%%%%%
\begin{abstract}
It is shown that negative Einstein metrics are attractors of the Einstein-Klein-Gordon system. As an essential part of the proof we upgrade a technique that uses the continuity equation complementary to $L^2$-estimates to control massive matter fields. In contrast to earlier applications of this idea we require a correction to the energy density to obtain sufficiently good pointwise bounds.
\end{abstract}

%%%%%%%%%%%%%%%%%%%%%%%%%%%%%%%%%%%%%%%%%%%%%%%%%%%%%
\section{Introduction}
Nonlinear stability results are milestones in the study of the Einstein vacuum equations. For two particular spacetimes, nonlinear stability for the vacuum Einstein flow is known, those are Minkowski spacetime \cite{CK93} and the Milne model \cite{AnMo11}. If the cosmological constant is non-vanishing a large class of deSitter type universes and black holes are known to be stable \cite{F,Ri08,HiVa18} and also towards the singularity of certain cosmological solutions, stability has been established recently \cite{RoSp18}. 
Restricting to stability results towards a complete direction of spacetime and the case of vanishing cosmological constant, i.e.~to the Milne model and Minkowski space, several results appeared recently that generalize these works to the non-vacuum setting. The matter models that have been considered in these generalizations include Maxwell fields \cite{BZ,BFK,CCL,S,Wy}, collisionless matter \cite{AF17,Faa,FJS17-2,LT,Ta16} and scalar fields \cite{LR,LM,Wj18,Wq18}, in particular Klein-Gordon fields, which are the subject of this paper. 

\subsection{The Einstein-Klein-Gordon system}
The Einstein-Klein-Gordon system (EKGS) describes nonvacuum spacetimes with massive scalar field(s) as the matter model. The system emerges as a projection of the Einstein equations on massive modes of the Fourier expansion of five-dimensional Kaluza-Klein metrics \cite{Wj18}. Independently, the system poses an interesting mathematical model as it consists of a system of quasilinear massless and massive wave equations. The presence of massive waves in the system makes its treatment substantially more difficult than that of, for instance, the Einstein-Scalar field system where only a massless wave is coupled to the Einstein equations. The present paper considers the stability problem of the Milne model in the expanding direction as a solution to the Einstein-Klein-Gordon system. 

The EKGS has been studied intensely in recent years with an emphasis on the nonlinear stability problem \cite{LM,Wq18,Wj18} (cf.~also \cite{HS13} where the decay of the KG field on a Kerr-AdS background is investigated).

\subsection{Cosmological spacetimes and stability}
The Milne model 
$((0, \infty) \times M, \bar g)$ with 
\eq{ \label{eq:background}
\bar{g} = - dt^2 + \tfrac{t^2}{9} \gamma,
}
where $M$ is a closed 3-manifold admitting an Einstein metric $\gamma$ with negative curvature, i.e. $R_{ij}[\gamma] = -\tfrac29 \gamma_{ij}$, is a solution to the vacuum Einstein equations in 3+1 dimensions. It models a universe emanating from a big bang singularity at $t=0$ expanding for all time with a linear scale factor. It is a member of the FLRW family of cosmologies and, in comparison with related isotropically expanding models for the Einstein equations with a positive cosmological constant (such as deSitter space), has the slowest expansion rate. This feature makes it difficult to establish stability results for the Milne model as decay rates of fields in cosmological spacetimes correspond inversely proportional to the rate of expansion.

The nonlinear stability of the Milne model in the expanding direction is known due to a series of works by Andersson and Moncrief who resolved this problem in general dimensions \cite{AnMo11}. Their approach uses the CMCSH gauge, developed in \cite{AnMo03}, which casts the Einstein equations into an elliptic-hyperbolic system. This gauge enables a crucial decomposition of the spatial Ricci tensor into an elliptic operator and perturbation terms. A corrected $L^2$-energy based on this operator in combination with control of its kernel then allows for a sufficiently strong energy estimate that yields decay of perturbations at a rate that implies future completeness of the spacetime. 

Recently, the stability of the Milne model has been generalized to the presence of a variety of matter models such as collisionless matter \cite{AF17} by Andersson and the first author, fields emanating from generalized Kaluza-Klein spacetimes, in particular electromagnetic fields \cite{BFK} by Branding, the first author and Kr\"oncke and also Klein-Gordon fields \cite{Wj18} by Wang. While the former two works use the CMCSH gauge to control the evolution of perturbations, in \cite{Wj18} a CMC-vanishing-shift gauge and Bel-Robinson-type energies (cf.~\cite{AnMo04}) are used.

A crucial difficulty that arises for massive matter models coupled to the Einstein equations for data close to the Milne model results from the slow decay of the lapse gradient. This was first observed for the Einstein-Vlasov system in \cite{AF17}. When this rate of decay is naively bootstrapped and enters the equation of motion for the matter, the resulting loss of decay for the matter field is too strong to close the argument. A similar difficulty arises for the Klein-Gordon field \cite{Wj18}.
Therein, it is resolved by a hierarchy which is initiated at lowest order of regularity by an estimate for the Klein-Gordon field that has been proven in the work on Minkowski stability by LeFloch and Ma \cite{LM} but surprisingly applies in the cosmological setting as well. 

\subsection{Upgrading decay estimates for massive fields by the continuity equation}
The main motivation for the present paper is a rough similarity between the Milne stability problem for the EKGS and the corresponding one for the massive Einstein-Vlasov system considered in \cite{AF17}. 
Therein, the crucial step to overcome the problem of slow decay of the lapse was the utilisation of the continuity equation, which turns into a first order evolution equation for the energy density. This evolution equation has a beneficial structure that allows one to obtain better estimates for the energy density than for a generic component of the energy-momentum tensor. As it is precisely the energy density that causes the lapse gradient to lose decay, this auxiliary estimate is the essential tool to obtain sharp bounds on the lapse and close the bootstrap argument. In \cite{AF17} is was conjectured that the continuity equation has similarly powerful applications in corresponding \emph{massive} matter models.\\
In the present paper we show that this conjecture is true for the EKGS. In particular, we use the continuity equation to obtain improved bounds for a suitable \emph{corrected energy density} on lowest order and based on this initialization construct a hierarchy of estimates increasing in regularity. In particular, we consider the rescaled energy density $\rho$ (defined in \eqref{resc-matt}) and correct it with a small indefinite term to obtain the corrected energy density
\eq{
\hat{\rho}=\rho-\frac12 \tau^2 \phi\left(\frac32N^{-1}\phi-\phi'\right),
}
where $\tau$, $\phi$ and $\phi'$ are mean-curvature, Klein-Gordon field and its time derivative defined in Section 2. The corrected energy density fulfils an evolution equation, given in \eqref{mod-cont-eq}, with only time-integrable terms on the right-hand side yielding uniform pointwise bounds on the energy density and, in turn, for the Klein-Gordon field. This approach is sufficiently strong to close a bootstrap argument for the full system. In contrast to \cite{Wj18}, we work in CMCSH gauge as the shift vector does not pose an obstruction to our approach, particularly when using the continuity equation. As a consequence, we have access to the concise energy-method to control the perturbation of the geometry and obtain a substantially shorter proof of the nonlinear stability problem. 

\subsection{Main theorem}
We formulate the main theorem using terminology introduced in Section \ref{sec : 2}. Let $\mathscr{B}^{k,l,m}_{\varepsilon}(\gamma,\frac13\gamma,0,0)$ denote the ball of radius $\varepsilon$ in the space $H^k\times H^l\times H^m$ centered at $(\gamma,\frac13\gamma,0,0)$.

\begin{thm}\label{thm-1}
Let $(M,\gamma)$ be a negative, closed 3-dimensional Einstein manifold with Einstein constant $\mu=-2/9$.
Let $\varepsilon>0$ and $(g_0,k_0,\phi_0,\dot\phi_0)$ be rescaled initial data for the Einstein-Klein-Gordon system at $\tau=\tau_0$ such that 
\eq{
(g_0,k_0,\phi_0,\dot\phi_0)\in \mathscr B_\varepsilon^{5,4,5,4}(\gamma,\frac13 \gamma,0,0).
}
Then, for $\varepsilon$ sufficiently small the corresponding future development under the Einstein-Klein-Gordon system is future complete and the rescaled metric and second fundamental form converge to
\eq{
(g,k)\rightarrow (\eta,\frac13 \eta) \mbox{ as } \tau\nearrow 0.
}
\end{thm}

\subsection{Overview on the paper}

The paper is organized as follows. Section 2 introduces notations and the fundamental equations. Section 3 discusses the $L^2$-energies for the Klein-Gordon field. Section 4 states the bootstrap assumptions and introduces the $L^2$-energies for the perturbation of the geometry. Section 5 discusses the continuity equation and its modification for the Klein-Gordon field. In Section 6 the energy estimates for the Klein-Gordon field are performed. 
Section 7 recalls the elliptic estimates for lapse and shift. Section 8 discusses the hierarchy of decay for the lapse function and the KG field. Finally, Section 9 closes estimates for the shift and the perturbation of the geometry and ends the proof.

\subsection*{Acknowledgements}
The authors acknowledge support of the Austrian Science Fund (FWF) through the Project \emph{Geometric transport equations and the non-vacuum Einstein flow} (P 29900-N27).

%%%%%%%%%%%%%%%%%%%%%%%%%%%%%%%%%%%%%%%%%%%%%%%%%%%%%
\section{Preliminaries} \label{sec : 2}
\subsection{The Einstein-Klein-Gordon system}
We consider the Einstein-Klein-Gordon system (EKGS) consisting of the Einstein equations
\eq{
\bar R_{\mu \nu} - \tfrac{1}{2} \bar R \bar g_{\mu \nu} = \tilde T_{\mu \nu}(\tilde \phi) 
}
where the stress-energy tensor is given by
\eq{ \label{eq:stressET}
\tilde T_{\mu \nu} (\tilde \phi) =  \bar \nabla_\mu \tphi \bar \nabla_\nu \tphi - \tfrac12 \bar g_{\mu\nu} \left( \bar g^{\rho \sigma} \bar \nabla_\rho \tphi \bar \nabla_\sigma \tphi + m^2 \tphi^2 \right) .
}
Note here $\bar \nabla$ is the Levi-Civita with respect to $ g$ and for $m>0$ the Klein-Gordon equation is
\eq{
%\Box_{\overline g}\, \tilde{\phi} = 
\bar \nabla^\mu \bar \nabla_\mu \tilde \phi = m^2\tilde{\phi}. 
}

\subsection{Negative Einstein metrics, Gauge choice and variables} The following setup is similar to earlier papers on the vacuum case or different matter models. We recall it briefly for the sake of completeness.
Throughout the paper let $\gamma$ denote a fixed negative Einstein metric such that $R_{ij}[\gamma]=-\frac29\gamma_{ij}$. We choose the constant for convenience, but any negative Einstein metric can be treated in the same way. Note also that Roman letters will always range over spatial indices $1,2,3$ and Greek letters will range over spacetime indices $0, 1,2,3$. 

To model the dynamic spacetime we consider the 3+1-dimensional metric in ADM form 
\eq{ \overline g=-\tilde N^2dt^2+\tilde g_{ab}(dx^a+\tilde X^adt)(dx^b+\tilde X^bdt)}
where $\tilde N$, $\tilde g$ and $\tilde X$ denote the lapse function, the induced Riemannian metric on $M$ and the shift vector field respectively\footnote{Note the coordinate $t$ here does not coincide with the same symbol in the explicit Milne model \eqref{eq:background}.}. Recall that $\bar g_{ab} = \tilde g_{ab}$ but in general $\bar g^{ab} \neq \tilde g^{ab}$. We denote by $\tau$ the trace of the second fundamental form $\tilde k$ with respect to $\tilde g$ and decompose $\tilde k=\tilde\Sigma+\tfrac13 \tau \tilde g$. We then impose the CMCSH gauge via
\eq{\label{gauges}
\begin{split}
t=\tau,\qquad
\widetilde g^{ij}(\widetilde\Gamma^{a}_{ij}-\widehat\Gamma^a_{ij})&=0,
\end{split}
}
where $\widetilde\Gamma$ and $\widehat\Gamma$ denote the Christoffel symbols of $\widetilde g$ and $\gamma$, respectively. 

\subsection*{Rescaling}
We rescale the variables $(\tilde g, \tilde \Sigma, \tilde N, \tilde X, \tilde \phi)$ with respect to mean curvature time $t=\tau$, calling the rescaled variables $(g, \Sigma, N, X, \phi)$. This coincides with earlier works, except for the Klein-Gordon field, which is rescaled here as follows. The rescaling is done according to
\eq{\label{rescaling}
\begin{array}{cc}
g_{ij}=\tau^2\tilde{g}_{ij}&N=\tau^2\tilde{N}\\
g^{ij}=\tau^{-2}\tilde{g}^{ij}& \Sigma_{ij}=\tau\tilde \Sigma_{ij}\\
\phi=-| \tau |^{-3/2}\tilde{\phi}& X^i=\tau \tilde{X}^i.
\end{array}
}
On top of this we introduce the logarithmic time $T=-\ln(\tau/\tau_0)$, ($\leftrightarrow\tau=\tau_0\exp(-T)$) with $\p_T=-\tau \p_\tau$. We have the following ranges $\tau_0 \leq \tau\nearrow 0$ and $0 \leq T \nearrow \infty$. 
We also define $\widehat N:=\frac N3-1$.
% and $\widehat X=X/N$. 
The rescaled Einstein equations in CMCSH gauge take the following form.
\begin{align}
R(g)-|\Sigma|_g^2+\tfrac{2}{3}&=4 
	\tau \rho \label{eq:EoM-R}
\\
\nabla^a \Sigma_{ab} &=
	\tau^2 \jmath_b \label{eq:EoM-Sigma}
\\
(\Delta - \tfrac{1}{3})N &= 
	N \left( |\Sigma|_g^2 - \tau \eta \right) \label{eq:EoM-lapse}
\\
\Delta X^i + R^i_m X^m &=
	2 \nabla_j N \Sigma^{ji} - \nabla^i \hN+ 2 N \tau^2 \jmath^i \label{eq:EoM-shift}
\\
& \quad \notag
	- (2N \Sigma^{mn} - \nabla^m X^n)(\Gamma^i_{mn} - \hat \Gamma^i_{mn})  
\\
\p_T g_{ab} &=
	2N \Sigma_{ab} + 2\hN g_{ab} - \mcr{L}_X g_{ab}
\\
\p_T \Sigma_{ab} &=
	-2\Sigma_{ab} - N(R_{ab} +\tfrac{2}{9}g_{ab} ) + \nabla_a \nabla_b N + 2N \Sigma_{ai} \Sigma^i_b
\\
& \quad \notag
	-\tfrac{1}{3} \hN  g_{ab} - \hN \Sigma_{ab} - \mcr{L}_X \Sigma_{ab} + N \tau S_{ab} .
\end{align}
Note here and throughout $\nabla$ denotes the Levi-Civita connection for the rescaled 3-metric $g$.  The energy density and energy current are defined by
\eq{
\tilde \rho = \tilde N^2 \tilde T^{00} \,, \,
\tilde \jmath_a = \tilde N \tilde T^0_a,
}
respectively. Furthermore $\tilde \jmath^a$ is defined by raising the index using $\tilde g^{ab}$, and in general we raise and lower indices on un-rescaled quantities (i.e. those with tildes) using $\tilde g$. We define the rescaled matter quantities as
\begin{equation}\label{resc-matt}
\begin{split}
\rho &:= \tilde \rho (-\tau)^{-3},\qquad\quad\,
\eta := (\tilde \rho + \tilde g^{ab} \tilde T_{ab} ) (-\tau)^{-3}, 
\\
\jmath^b &:= \tilde \jmath^b (-\tau)^{-5},\qquad
S_{ab} := (\tilde T_{ab} - \tfrac12 \tilde g_{ab} \tilde T ) (-\tau)^{-1} .
\end{split}
\end{equation}
Here and throughout, spatial indices for rescaled quantities will be raised and lowered using the rescaled metric $g$. Thus $ \jmath_a$ is defined using the rescaled metric $g_{ab}$, and moreover $\jmath_a$ would scale as $\tilde \jmath_a (-\tau)^{-3}$. 
For the Klein-Gordon field these matter quantities are evaluated as
\begin{align}
\rho &= \tfrac12 m^2 \phi^2 + \tfrac{\tau^2}{2} \big( \tfrac32 N^{-1} \phi - \phi' \big)^2 + \tfrac12 \tau^2 g^{ab} \nabla_a \phi \nabla_b \phi  \label{eq:rho}
\\
\jmath^a &= \tau \big( \tfrac32 N^{-1} \phi - \phi' \big) g^{ab} \nabla_b \phi \label{eq:j}
\\
\eta &= -\tfrac12 m^2 \phi^2 + 2 \tau^2 \big( \tfrac32 N^{-1} \phi - \phi' \big) ^2 
\label{eq:eta}
\\
S_{ab} &= \tfrac12 m^2 \phi^2 g_{ab} + \tau^2 \nabla_a \phi \nabla_b \phi \label{eq:Sij}
\\
g_{ab} T^{ab} &=
- \tfrac32 \tau^{-2} m^2 \phi^2 + \tfrac32  \big( \tfrac32 N^{-1} \phi - \phi' \big) ^2
- \tfrac12 g^{ab} \nabla_a \phi \nabla_b \phi, \label{eq:gab-Tab}
\end{align}
where we have used the following notation 
\eq{
\hdel := \p_T + \mathcal{L}_X \,, \quad  \phi' := N^{-1} \hdel \phi .
}
For completeness note that the unit normal vector acting on $\tilde \phi$ becomes, in the rescaled variables,
\eq{
\tilde N^{-1} (\p_\tau - \tilde X^a \p_a) \tilde \phi
= \tau^2 N^{-1} \cdot \tau^{-1} \hdel \big( (-\tau)^{3/2} \phi \big)
= (- \tau)^{5/2} \big( \tfrac32 N^{-1} \phi - \phi'\big).
}
Finally, the rescaled Klein-Gordon equation takes the form
\eq{\label{eq:rescaledKG-eq}
 \hdel  \phi' = \nabla^a (N \nabla_a \phi ) +(4-N) \phi' +\tfrac32 \phi - \tfrac{15}{4} N^{-1} \phi - \tfrac32 N^{-2} \phi \hdel N - \tau^{-2} m^2 N \phi .
}
Note to derive \eqref{eq:rescaledKG-eq} it was convenient to move to a `Cauchy adapted frame', see for example \cite[VI\S 3]{CB09} 
This ends the setup of the EKGS in the CMCSH gauge with appropriate rescaling. In the following we will work solely with these equations.
%%%%%%%%%%%%%%%%%%%%%%%%%%%%%%%%%%%%%%%%%%%%%%%%
\section{Energy functionals for the Klein-Gordon field}
In this section we define the $L^2$-energy of the Klein-Gordon field in two steps. First, we define the natural $L^2$-norm of a massive scalar field. In the second step we modify this energy with two non-definite terms to obtain the corrected energy, which turns out to fulfil the desired energy estimate, which is derived later.

\subsection{Natural energy} 
The following energy is the natural $L^2$-energy expressed in the rescaled variables. 
%(cf.~\cite{Wj18}, (4.23)). 
\begin{defn}
\eq{ \begin{aligned}
E_k(\phi) 
&:=
	\int_\Sigma \tau^2 (-1)^k \left(
	\phi' \Delta^k \phi' - \phi \Delta^{k+1} \phi 
	\right) \sqrt{g} 
	+ \int_\Sigma m^2 (-1)^k \phi \Delta^k \phi \sqrt{g},
\\
\Ephi{\ell} 
&:= 
	\sum_{k=1}^\ell E_k(\phi).
\end{aligned}
}
\end{defn}
%We clearly have $E_k(\phi) \geq 0$ since under the integral
%$$
%\phi' \Delta^k \phi'	- 3 \Delta^k (N^{-1} \phi) \phi' + \tfrac94 N^{-1} \phi \Delta^k (N^{-1} \phi) 
%= \left( \phi' - \tfrac32 N^{-1} \phi \right) \Delta^k \left( \phi' - \tfrac32 N^{-1} \phi \right)
%$$
We need the following lemma further below.
\begin{lem} \label{lem:L2norm-laplacian-equiv}
The following equivalence (denoted $\cong$) holds
\eq{
\| \phi \|_{H^{k+2}} \cong \| \Delta \phi \|_{H^k} + \| \phi \|_{L^2}
}
for a sufficiently regular function $\phi$. This implies
\eq{
\| \phi \|_{H^k}
 \cong
	\| \Delta^\fk \phi \|_{L^2} + \| \nabla^\oddk \Delta^\fk \phi \|_{L^2} + \| \phi \|_{L^2},
}
where $\oddk=\ck - \fk$.
Thus 
\eq{\label{eq:Ephi-equiv}
\Ephi{k}^{1/2} \cong \| \tau \phi'\|_{H^k} + \|\tau \phi \|_{H^{k+1}} + \| m \phi \|_{H^k}.
}
\end{lem}

\begin{proof}
This follows from \cite[Appendix H, Theorem 27]{Be08}. 
%The last part explicitly for $k$ even and $k$ odd:
%\begin{align*}
%\| \phi \|_{H^{2l}} 
%&\cong 
%	\| \Delta^l \phi\|_{L^2} + \| \phi \|_{L^2} 
%\\
%\| \phi \|_{H^{2l+1}}
%& \cong
%	\| \Delta^l \phi \|_{H^1} + \| \phi \|_{L^2} = \| \Delta^l \phi \|_{L^2} + \| \nabla \Delta^l \phi \|_{L^2} + \| \phi \|_{L^2}
%\end{align*}
\end{proof}
It is important to note the appearance of $|\tau|$ weights in \eqref{eq:Ephi-equiv}.

%%%%%%%%%%%%%%%%%%%%%%%%%%%%%%%%%%%%%%%%%%%%%%%%%%%%%
\subsection{Modified energy} \label{subsec:modified-energy-KG}
Now we introduce the modified energy, which contains two indefinite terms (cf.~\cite{Wj18}, (4.24)). Its equivalence to the standard energy is shown below.
\begin{defn}
\eq{\label{eq:modified-energy-KG}
\begin{aligned}
\tilde E_k(\phi) 
 &:=
	\int_\Sigma \tau^2 (-1)^k \Big[ 
	\phi' \Delta^k \phi' - \phi \Delta^{k+1} \phi 
+ 3 \phi'\Delta^k \big( N^{-1} \phi \hN \big) 
\\
& \qquad
	+ \tfrac32 N^{-1}\phi \Delta^k \left( (\tfrac32-N) N^{-1} \phi \right) 
	\Big] \sqrt{g}
	+ \int_\Sigma m^2 (-1)^k \phi \Delta^k \phi \sqrt{g},
\\
\tEphi{\ell} 
&:=
	\sum_{k=0}^\ell \tilde E_k(\phi).
\end{aligned}
}
\end{defn}

\begin{lem} \label{lem:E-equivalence}
Assume that there exists a constant $C>0$ such that $\|{N}\|_{L^\infty}+\|N^{-1}\|_{L^\infty}+\|g-\gamma\|_{H^N}<C$ and $\|\widehat N\|_{H^{N+1}} \leq C e^{-\tfrac12 T}$. Then there exists a $\tau_0$ such that for all $\tau>\tau_0$ the following equivalence holds.
\eq{
\tEphi{\ell}\cong \Ephi{\ell}\,, \quad \ell \leq N .
}
\end{lem}

\begin{proof}
We first write the difference between the two energies (note without summing in $k$). 
\eq{ \label{eq:EminusEt}
\begin{aligned}
E_k(\phi) - \tilde E_k(\phi)  
&=
	\int_\Sigma (-1)^k\tau^2 \Big( 
	3 \phi'  \Delta^k (N^{-1} \phi) - \tfrac94  N^{-1} \phi \Delta^k (N^{-1} \phi) \Big) \sqrt{g}
\\
& \quad
	+  \int_\Sigma (-1)^k \tau^2  \Big( \tfrac32 N^{-1} \phi \Delta^{k}\phi - \phi' \Delta^{k} \phi \Big) \sqrt{g}.
\end{aligned}
}
Examine the first term on the right hand side of \eqref{eq:EminusEt}. The claim for $\ell=0$ is easily seen from:
\begin{align*}
\left| \int_\Sigma \tau^2 \left( 
	-3 (N^{-1} \phi) \phi' 
	\right) \sqrt{g} \right|
\leq
	C \delta \int_\Sigma \tau^2 (\phi')^2 \sqrt{g}
	+ {C \tau_0^2 \over \delta m^2} 
	\int_\Sigma m^2 \phi^2 \sqrt{g} 
\end{align*} 
where we used $\| N^{-1} \|_{L^\infty} \leq C$ and $\delta$ is a constant we are free to choose. For sufficiently small $\delta$ and $\tau_0 = \text{tr}\tilde k_0$ one can ensure all coefficients are strictly less than 1. Note that $0>\tau>\tau_0$ implies $|\tau|<|\tau_0|$. Thus this term can be absorbed by the clearly positive terms of $E_0(\phi)$.

A similar argument holds for $\ell \geq 1$.
For some smooth functions $v,w$ 
$$
\Delta^i(vw) = (\Delta^i v) w + \sum_{|I|+|J|+1=2i} c_{IJ} \nabla^{I} v \nabla^{J+1} w
$$
for some coefficients $c_{IJ}$ depending on $g$. 
So for a fixed value of $k \geq 1$, integration by parts (where there are $2\fk+\oddk = k$  derivatives distributed) gives
\begin{align*}
&\left| \int_\Sigma \tau^2 (-1)^k \left( 
	3 \Delta^k (N^{-1} \phi) \phi' 
	\right) \sqrt{g} \right|
\\
& \quad =	
	3 \left| \int_\Sigma \tau^2  \left( 
	 (\nabla_a)^\oddk \Delta^\fk (N^{-1} \phi) \right) \left( (\nabla^a)^\oddk \Delta^\fk \phi' 
	\right) \sqrt{g} \right|
\\
& \quad \leq
	C \int \tau^2 N^{-1} \big| \nabla^\oddk \Delta^\fk \phi \nabla^\oddk \Delta^\fk \phi' \big| \sqrt{g}
	+ C  \sum_{|I|+|J| +1=k} \int \tau^2 \big| \nabla^I \phi \nabla^{J+1} N^{-1} \nabla^\oddk \Delta^\fk \big| \sqrt{g}
\\
& \quad \leq	
	C \| N^{-1} \|_{L^\infty} \| \tau \nabla^\oddk \Delta^\fk \phi \|_{L^2} \| \tau \nabla^\oddk \Delta^\fk \phi' \|_{L^2}
	+ C \sum_{|I|+|J| +1=k} \| \tau \nabla^I \phi \|_{L^4} \| \nabla^{J+1} N^{-1} \|_{L^4}
	\| \tau \nabla^\oddk \Delta^\fk \phi' \|_{L^2}^2
\\
& \quad \leq	
	C   \| \tau \phi\|_{H^k} \| \tau \phi'\|_{H^k}
	+ C \sum_{|I|\leq k-1} \| \tau \nabla^I \phi \|_{H^1} \cdot
	\sum_{1 \leq |J| \leq k-1} \| \nabla^J N^{-1} \|_{H^1}
	\cdot \| \tau \phi'\|_{H^k}
\\
& \quad \leq	
	C   \| \tau \phi\|_{H^k} \| \tau \phi'\|_{H^k}
	+ C\| \tau \phi \|_{H^k} \cdot
	\| \hN \|_{H^{k+1}} \| \tau \phi'\|_{H^k}
\\
& \quad \leq	
	C \delta  \| \tau \phi'\|_{H^k} ^2 +  {C \tau_0^2 \over \delta m^2}  \| m \phi\|_{H^k} ^2.
\end{align*} 
In the third to last line we used the following estimate
\eq{ \label{eq:nablaN-inverse} \begin{aligned}
\sum_{1 \leq |J| \leq k} \| \nabla^J N^{-1} \|_{L^2}
%& \leq
%	\sum_{\sum |\alpha_i| = k } \|\mathcal{P}_k\big(N^{-1};  \nabla^{\alpha_1} \hN \cdots \nabla^{\alpha_k} \hN \big) \|_{L^p}
%\\
& \leq	
	\| N^{-2} \|_{L^\infty} \| \hN \|_{H^k} 
	+ 	C( \| N^{-1} \|_{L^\infty}) \sum_{|I| \leq \fk} \| \nabla^I \hN \|_{L^\infty} \| \hN \|_{H^k} 
\\
& \leq	
	C \| \hN \|_{H^k} 
	+ C \| \hN \|_{H^{k+1}} \| \hN \|_{H^k} 
 \leq	
	C \| \hN \|_{H^k} .
\end{aligned} }
%where $\mathcal{P}_k(N^{-1}; \nabla^\alpha \hN)$ indicates an arbitrary polynomial in derivatives $\nabla^\alpha \hN$ with coefficients containing (positive) powers of $N^{-1}$. 
In this estimate we used $\fk+2 \leq k+1$ for $k \geq 1$ in order to take the terms with low derivatives on $\hN$ out in $L^\infty$ and embed using Sobolev, recalling also that $N$ is controlled at one order of regularity higher than $\phi$.  Note we also used the Sobolev-embedding $H^1 \hookrightarrow L^4$. 

In a similar way we can show
\begin{align*}
 \Big| \int_\Sigma & \tau^2  N^{-1} \phi \Delta^k (N^{-1} \phi) \sqrt{g} \Big|
+ \Big| \int_\Sigma \tau^2  \Big( \phi' \Delta^{k} \phi - \tfrac32 N^{-1} \phi \Delta^{k}\phi \Big) \sqrt{g} \Big|
\\ 
& \leq
	\int_\Sigma \big| \tau \nabla^\oddk \Delta^\fk (N^{-1} \phi) \big|^2 \sqrt{g}
	+ C \int_\Sigma \tau^2 \left| \nabla^\oddk \Delta^\fk (N^{-1} \phi) \nabla^\oddk \Delta^\fk \phi \right| \sqrt{g}
\\
& \quad
		+ 	C \delta \int_\Sigma \tau^2  |\nabla^\oddk \Delta^\fk  \phi'|^2 \sqrt{g}
	+ {C \tau_0^2 \over \delta m^2} \int_\Sigma m^2 | \nabla^\oddk \Delta^\fk \phi|^2 \sqrt{g}
\\ 
& \leq	
	{C \tau_0^2 \over m^2} \| m \phi \|_{H^k}^2 
	+ \| \hN \|_{H^k}^2 \| \tau \phi \|_{H^{k+1}}^2
	+ C \delta \| \tau \phi'\|_{H^k}^2 
	+ {C \tau_0^2 \over \delta m^2} \| m \phi \|_{H^k}^2
\\ 
& \leq	
	{C \tau_0^2 \over m^2} \Ephi{k}
	+ |\tau_0 | \Ephi{k}
	+ C \delta \Ephi{k}
	+ {C \tau_0^2 \over \delta m^2} \Ephi{k}.
\end{align*}
Thus the claim holds by summing in $k$ from $0$ to $\ell$ and reducing $\delta$ and $|\tau_0|$ to be sufficiently small so that all coefficients can be made to be strictly smaller than 1.
\end{proof}

%%%%%%%%%%%%%%%%%%%%%%%%%%%%%%%%%%%%%%%%%%%%%%%%%%%%%
\section{Energy norms and smallness assumptions}
In this section we state the global bootstrap assumptions to provide for a simpler notation in all the estimates to follow. Also, we recall the definition of the corrected $L^2$-energy to control the perturbation of the geometry from the earlier works \cite{AnMo11,AF17}.
\subsection{Bootstrap assumptions}
Fix the regularity $N \geq 4$ and some constant $0 < \gamma << 1$. Then we impose the following assumptions on the data we consider.	
\begin{subequations} \label{eq:bootstraps}
\begin{align}
\| g- \gamma\|_{H^{N+1}}^2 + \| \Sigma \|_{H^{N}}^2 & \leq C_I \varepsilon^2 e^{-\tfrac32 T},
\\
\| \hN \|_{H^{N+1}} &\leq C_I \varepsilon e^{(-1+\gamma)T},
\\
\| X \|_{H^{N+1}} & \leq C_I \varepsilon e^{(-1+\gamma)T},
\\
\mathcal{E}_{N}(\phi) & \leq C_I \varepsilon^2 e^{2 \gamma T} .
\end{align}
\end{subequations}
A simple check shows the assumptions of Lemma \ref{lem:E-equivalence} are consistent with \eqref{eq:bootstraps}. 
\subsection{Energy for the perturbation of the geometry} \label{subsec:energy-geom-def}
As in \cite{AnMo11} and later in \cite{AF17} define an energy for the geometric perturbation as follows. We define the correction constant $\alpha=\alpha(\lambda_0,\delta_\alpha)$ by
\eq{
\alpha=
\begin{cases}
1& \lambda_0>1/9\\
1-\delta_\alpha& \lambda_0=1/9,
\end{cases}
}
where $\delta_\alpha=\sqrt{1-9(\lambda_0-\varepsilon')}$ with $1>>\varepsilon'>0$ remains a variable to be determined in the course of the argument to follow. By fixing $\varepsilon'$ once and for all, $\delta_\alpha$ can be made suitable small when necessary.\\
The corresponding correction constant, relevant for defining the corrected energies is defined by
\eq{
c_E=\begin{cases}
1& \lambda_0>1/9\\
9(\lambda_0-\varepsilon')& \lambda_0=1/9.
\end{cases}
}
We are now ready to define the energy for the geometric perturbation. For $m\geq 1$ let
\eq{\begin{split}
\mathcal{E}_{(m)}&=\frac12\int_M\langle 6\Sigma,\mathcal{L}_{g,\gamma}^{m-1} 6\Sigma \rangle\mu_g+\frac92\int_M \langle (g-\gamma),\mathcal{L}_{g,\gamma}^{m}(g-\gamma)\rangle\mu_g\\
\Gamma_{(m)}&=\int_M \langle 6\Sigma,\mathcal{L}_{g,\gamma}^{m-1}(g-\gamma)\rangle\mu_g.
\end{split}}

Then, the following corrected energy for the geometric perturbation is defined by
\begin{defn}
\eq{
E_k(g, \Sigma)= \sum_{1\leq m\leq k} \mathcal{E}_{(m)}+c_E\Gamma_{(m)}.
}
\end{defn}
Under the imposed conditions, the energy is coercive and equivalent
$$
\Eg_k(g, \Sigma) \cong \| g- \gamma\|_{H^{k+1}}^2 + \| \Sigma \|_{H^k}^2 .
$$

\subsection{Local Well-posedness}
Local existence theory is a prerequisite for addressing the global existence and stability problem for any Einstein-matter system. The local existence problem for the vacuum Einstein equations in CMCSH gauge was proven in \cite{AnMo03}. We provide the corresponding result for the Einstein-Klein-Gordon system below. As it differs from the vacuum system only by coupling an additional nonlinear wave equation to the elliptic-hyperbolic system there is essentially no difference in the proof. One issue is however important to remark, which concerns the elliptic system. To preserve the crucial feature that the elliptic operators are isomorphisms we need to impose a smallness condition on the matter variables. This has been observed already in \cite{Fa16} for collisionless matter and, for simplicity, turned into a smallness assumption for the full perturbation. Following the strategy of proof in \cite{AnMo03} and making an  additional smallness assumption analogously to that in \cite{Fa16} yields the following local-existence theorem for the EKGS.
\begin{lem}
There exists a $\delta>0$ such that for any initial data set at time $T_0$ for the rescaled Einstein-Klein-Gordon system in CMCSH gauge $(g_0,k_0,\phi_0,\dot\phi_0)$ with
\eq{
\|g_0-\gamma\|_5+\|\Sigma_0\|_4+\sqrt{|\tau|}(\|{\phi_0}\|_5+\|\dot{\phi_0}\|_4)\leq \delta ,
}
there exists a local-in-time solution to the rescaled EKGS in CMCSH gauge with this initial data. Moreover, let $T_+$ be the supremum of all $T>T_0$ such that the corresponding solution exists up until $T$. Then either $T_+=\infty$ or 
\eq{
\limsup_{T\rightarrow T_+}\|g-\gamma\|_5+\|\Sigma\|_4+\sqrt{|\tau|}(\|{\phi}\|_5+\|\dot\phi\|_4)\geq2\delta.
} 
\end{lem}

\begin{rem}
The mean-curvature factor in front of the Klein-Gordon field terms results from the lapse equation, where this condition assures smallness of the corresponding matter term $-\tau\eta$. As $\eta$ is quadratic in the rescaled Klein-Gordon field, each terms obtains a factor of $\sqrt{|\tau|}$.
\end{rem}

\begin{rem}
We use the previous lemma to establish global existence of the solution corresponding to the considered perturbations by proving decay estimates that assure, in particular, that the necessary bounds above hold.
\end{rem}

%%%%%%%%%%%%%%%%%%%%%%%%%%%%%%%%%%%%%%%%%%%%%%%%%%%%%

\section{Modified continuity equation}
In this section we cast the rescaled Klein-Gordon equation in the particular form \eqref{eq:rescaledKG-eq2} and derive a modified continuity equation \eqref{mod-cont-eq}. This is essential to prove an improved bound for the pointwise norm of the energy-density, which in turn allow for an initialization of the hierarchy that improves bounds on the Klein-Gordon field and the lapse function.\\

The rescaled Klein-Gordon equation \eqref{eq:rescaledKG-eq} when written in terms of the small quantity $\hN$ reads
\eq{\label{eq:rescaledKG-eq2}
 \hdel  \phi' = \nabla^a (N \nabla_a \phi ) -3\hN \phi' +\phi' + \tfrac92 N^{-1} \phi \hN + \tfrac34 N^{-1} \phi + \tfrac32 \phi \hdel N^{-1} - \tau^{-2} m^2 N\phi.
}
The standard continuity equation, see for example \cite{Re08}, is
\begin{align*}
\p_T \rho
&=
	(3-N)\rho - X^a \nabla_a \rho 
	+ \tau N^{-1} \nabla_a (N^2 j^a) 
	- \tau^2 {\tfrac{N}{3}} g_{ab} T^{ab}
	- \tau^2 N \Sigma_{ab} T^{ab}
\\
&=
	- 3 \hN\rho - X^a \nabla_a \rho 	
	+ \tau^2 N^{-1} \nabla_a \Big( N \nabla^a \phi ({\tfrac32}\phi - \hdel \phi) \Big)
	+ {\tfrac12} N m^2 \phi^2 
\\
& \quad
	+ \tau^2 \tfrac N6 \nabla^a \phi \nabla_a \phi 
	- \tau^2 \tfrac{N^{-1}}{2} \Big({\tfrac32} \phi - \hdel \phi \Big)^2 + \tau^2 {\tfrac{N}{2}} \Sigma_{ab} \nabla^a \phi \nabla^b \phi.
\end{align*}
The problematic terms in this expression are $N m^2 \phi^2$ and $N ( \tau \phi')^2$. This is because naively estimating such terms using the standard Sobolev embedding $L^\infty \hookrightarrow H^2$ leads to an problematic $e^{\gamma T}$ growth for $\rho$. Nonetheless, motivated by the notion of a `modified' energy, we consider a modified energy density. Consider the quantity
\eq{
P:= 
	\phi \Big( \tfrac32 N^{-1} \phi - \phi' \Big).
}
Up to a factor $ \tau^2$, this is similar to one of the terms subtracted from $E_0(\phi)$ to obtain $ \tilde E_0(\phi)$. 
Using the KG equation \eqref{eq:rescaledKG-eq} its  evolution equation is the following. 
\begin{align*}
\p_T P
&=
	- \mathcal{L}_X P + \hdel P
\\
&=
	- \mathcal{L}_X P + 3 \phi \phi' + \tfrac32 \phi^2 \hdel N^{-1} - N(\phi')^2 - \phi \hdel ( \phi' )
\\
%&=
%	- \mathcal{L}_X P + 3 \phi \phi' + \tfrac32 \phi^2 \hdel N^{-1} - N(\phi')^2
%\\
%& \quad
%	- \phi \Big( \nabla^a(N \nabla_a \phi) - N \phi' + \tfrac32 \phi + 4 \phi' - \tfrac32 N^{-2} \phi \hdel N - \tfrac{15}{4} N^{-1} \phi - \tau^{-2} m^2 Nu \Big)
%\\
&=
	- \mathcal{L}_X P + 3 \phi \phi'  
	- \phi \nabla^a(N \nabla_a \phi) + N \phi'\phi -\tfrac32 \phi^2 - 4 \phi'\phi  + \tfrac{15}{4} N^{-1} \phi^2 
\\
& \quad
	- N(\phi')^2 + \tau^{-2} m^2 N \phi ^2.
\end{align*}
For some constant $\lambda$ define
\eq{
\hat \rho 
:= 
	\rho + \lambda \tau^2 P.
}
We will be interested in the modified continuity equation for $\lambda=-1/2$
\eq{\label{mod-cont-eq}\begin{aligned}
\p_T \hat \rho 
&=
	-3\hN \rho - X^a \nabla_a \hat \rho + \tau^2 {\tfrac12} (1+3/N) \phi \nabla_a N \nabla^a \phi + \tau^2 {\tfrac12} (3+N) \phi \Delta \phi
\\
& \quad
	- \tau^2 N \phi' \Delta \phi 
	+ \tau^2 \tfrac32 (1+N/9) \nabla_a \phi \nabla^a \phi
	- \tau^2 N \nabla^a \phi \nabla_a \phi'
	- 2 \tau^2 \phi'\nabla^a \phi \nabla_a N
\\
&\quad
	+ \tau^2 {\tfrac34} (1 - 2/N) \phi^2 
	+ \tau^2 {\tfrac12} (1-N/2) \phi \phi'
	+ \tau^2 {\tfrac{N}{2}} \Sigma_{ab} \nabla^a \phi \nabla^b \phi.
\end{aligned}
}
\begin{prop} \label{prop:modified-cont}
Assume the bootstrap assumptions \eqref{eq:bootstraps} hold. 
If $\lambda = -1/2 $ the following equivalence holds
$\rho \cong \hat \rho $
and also the estimates
\eq{
| \p_T \hat \rho | \lesssim
	\big( \| \hN \|_{H^3} 
	+ \| X \|_{H^2} 
	+ \| \Sigma \|_{H^2} 
	+ |\tau|
	\big) \Ephi{4},
	}
and thus
\eq{
\rho |_{T} \lesssim
	 \rho |_{T_0} 
	+ C \varepsilon^3 \int_{T_0}^T e^{(-1+2\gamma)s} ds
\leq 
	C \varepsilon^2.
}
\end{prop}

\begin{proof}

\begin{align*}
\p_T \hat \rho 
&=
	\p_T \rho + \lambda \tau^2 \partial_T P - 2 \lambda \tau^2 P
\\
& =
	-3 \hN \rho - X^a \nabla_a \hat \rho  
	+ N m^2 \phi^2 ( {\tfrac12} + \lambda) - \tau^2 N (\phi')^2 ( {\tfrac12} + \lambda)
\\
&\quad
	+ \tau^2 N^{-1} \nabla_a N \nabla^a \phi \Big( \tfrac32 \phi - \hdel \phi \Big)
	+ \tau^2  \Delta \phi \Big( \tfrac32 \phi - \hdel \phi \Big) 
	+ \tau^2 \nabla^a \phi \nabla_a \Big( \tfrac32 \phi - \hdel \phi \Big) 
\\
&\quad
	+ \tau^2 \tfrac{N}{6} \nabla_a \phi \nabla^a \phi 
	- \tau^2 \tfrac98 N^{-1} \phi^2 
	+ \tau^2 \tfrac32 \phi \phi'
	+ \tau^2 {\tfrac{N}{2}} \Sigma_{ab} \nabla^a \phi \nabla^b \phi
\\
&\quad
	+ \lambda \tau^2 \big( - \phi \phi' - \phi \nabla^a(N \nabla_a \phi) + N \phi \phi' - \tfrac32 \phi^2 + \tfrac{15}{4} N^{-1} \phi^2 - 3N^{-1} \phi^2 + 2 \phi \phi' \big).
\end{align*}
Choosing $\lambda = -1/2$ we can remove the problematic terms. Indeed the evolution equation is now
\eqref{mod-cont-eq}.
%\begin{align*}
%\p_T \hat \rho 
%&=
%	-3\hN \rho - X^a \nabla_a \hat \rho + \tau^2 {\tfrac12} (1+3/N) \phi \nabla_a N \nabla^a \phi + \tau^2 {\tfrac12} (3+N) \phi \Delta \phi
%\\
%& \quad
%	- \tau^2 N \phi' \Delta \phi 
%	+ \tau^2 \tfrac32 (1+N/9) \nabla_a \phi \nabla^a \phi
%	- \tau^2 N \nabla^a \phi \nabla_a \phi'
%	- 2 \tau^2 \phi'\nabla^a \phi \nabla_a N
%\\
%&\quad
%	+ \tau^2 {\tfrac34} (1 - 2/N) \phi^2 
%	+ \tau^2 {\tfrac12} (1-N/2) \phi \phi'
%	+ \tau^2 {\tfrac{N}{2}} \Sigma_{ab} \nabla^a \phi \nabla^b \phi
%\end{align*}
To show the equivalence, note
\eq{
| \rho - \hat \rho| = 
{\tau^2 \over 2} \Big|  \phi \Big( \tfrac32 N^{-1} \phi - \phi' \Big) \Big|
\leq 
	{\tau_0^2 \over \delta m^2} m^2 \phi^2
	+ \delta \tau^2\Big( \tfrac32 N^{-1} \phi - \phi' \Big)^2.
}
Thus for sufficiently small $\delta$ and $\tau \geq \tau_0$ equivalency holds. Finally, and recalling \eqref{eq:Ephi-equiv} and the standard Sobolev embedding $H^2 \hookrightarrow L^\infty$, we have
\begin{align*}
| \p_T \hat \rho |
& \lesssim
	\| \hN \|_{L^\infty} \big( \| m \phi \|_{L^\infty}^2 + \| \tau N^{-1} \phi \|_{L^\infty}^2 + \| \tau \phi' \|_{L^\infty}^2 \big)
	+ \| X \|_{L^\infty} 
	\| m^2 \phi \nabla \phi \|_{L^\infty} 
\\
& \quad
	+ \| X \|_{L^\infty} (\| \tau N^{-1} \phi \|_{L^\infty} + \| \tau \phi'\|_{L^\infty}) ( \| \tau \nabla (N^{-1} \phi )\|_{L^\infty} + \| \tau \nabla \phi' \|_{L^\infty} )
	+ \| \nabla \hN \|_{L^\infty} \| \tau^2 \phi \nabla \phi \|_{L^\infty}
\\
& \quad
	+|\tau| \| \tau \phi' \|_{L^\infty} \| \Delta \phi \|_{L^\infty}
	+ \tau^2 \| \nabla \phi \|_{L^\infty}^2 
	+ | \tau | \| \nabla \phi \|_{L^\infty} \| \tau \nabla \phi' \|_{L^\infty}
	+ \| \tau \phi' \|_{L^\infty} \| \tau \nabla \phi \|_{L^\infty} \| \nabla \hN \|_{L^\infty}
\\
& \quad
	+ \tau^2 \| \phi \|_{L^\infty}^2 
	+ | \tau | \| \tau \phi' \|_{L^\infty} \| \phi \|_{L^\infty}
	+ \| \Sigma \|_{L^\infty} \| \tau \nabla \phi \|_{L^\infty}^2
\\
& \lesssim 
	\| \hN \|_{H^3} \Ephi{3}
	+ \| X \|_{H^2}\Ephi{3}
	+ | \tau | \Ephi{4}
	+ \| \Sigma \|_{H^2} \Ephi{2} \,.
\end{align*}
Lastly we estimate the initial value of $\rho$ 
\begin{align*}
\rho \big|_{T_0} \lesssim
	\big( \| m \phi \|_{L^\infty}^2 + \| \tau N^{-1} \phi \|_{L^\infty}^2 + \| \tau \phi' \|_{L^\infty}^2 \big) \big|_{T_0}  
\lesssim
	\Ephi{2} \big|_{T_0} 
\lesssim \varepsilon^2 e^{4 \gamma T_0 } \,.
\end{align*}
\end{proof}

%%%%%%%%%%%%%%%%%%%%%%%%%%%%%%%%%%%%%%%%%%%%%%%%%%%%%
\section{Energy inequalities}
In this section we will derive decay inequalities for the time derivative $\p_T$ of the modified $L^2$-energy norm of the Klein-Gordon field defined in Section \ref{subsec:modified-energy-KG}. The main results in this section, Propositions \ref{prop:Ezero-KG} and \ref{prop:prelim-est-Ek-KG}, will then be combined with later estimates for the Lapse in Lemma \ref{lem:lapse-est} in order to close the bootstrap argument. 

\subsection{Zeroth order Klein-Gordon energy}
Recall the definition
\begin{align*}
\tEphi{0}
&= 
	\int_\Sigma \tau^2 \left[ 
	(\phi')^2 + g^{ab} \nabla_a \phi \nabla_b \phi + 3  N^{-1}\phi \phi' \hN
	+ \tfrac32  N^{-2} \phi^2\left( \tfrac32-N \right) \right] \sqrt{g}
	+ \int_\Sigma m^2 \phi^2 \sqrt{g}.
\end{align*}

\begin{prop} \label{prop:Ezero-KG}
Assume the bootstrap assumptions \eqref{eq:bootstraps}  hold. Then the zeroth-order estimate holds
\eq{
\p_T \tEphi{0}
 \lesssim 
		 \big( \| \Sigma \|_{H^2} + \| \hN \|_{H^3}
		 + | \tau | \big) \tEphi{0}
}
and thus 
\eq{
\tEphi{0} \big|_T
 \lesssim 
		\tEphi{0} \big|_{T_0} 
		\cdot \exp \Big( C \int_{T_0}^T  e^{(-1+\gamma)s}  ds \Big) \,.
}

\end{prop}
\begin{proof}
The modified energy takes the form 
$$ \tEphi{0} =: \int_\Sigma (\tau^2 f_0(\phi) + m^2 \phi^2) \sqrt{g}$$
where $f_0(\phi)$ is the expression between the square brackets above.  An identity taken from \cite[(6.5)]{AF17}, valid for some function $u$ on $\Sigma$, is the following
\eq{
\p_T \int_\Sigma u \mu_g = 3 \int_\Sigma \hN u \mu_g + \int_\Sigma \hdel(u) \mu_g \,.
}
Thus we find 
\begin{align}
\p_T \tEphi{0}
\notag &= - \int_\Sigma (3- N)(\tau^2 f_0(\phi) + m^2 \phi^2) \sqrt{g} + \int_\Sigma \left( \tau^2 \hdel(f_0) - 2 \tau^2 f_0 + m^2 \hdel(\phi^2) \right) \sqrt{g}
\\
& \leq \| \hN \|_{L^\infty} \tEphi{0} + \Big| \int_\Sigma \tau^2 \hdel(f) - 2 \tau^2 f + m^2 \hdel(\phi^2)  \sqrt{g} \Big| \,.
\label{eq:pT-Ephi}
\end{align}
Another identity taken from \cite[(6.4)]{AF17} and relevant for this and later calculations is
\eq{\label{eq:hdelg-up}
\hdel g^{ab} = -2N \Sigma^{ab} - 2 \hN g^{ab} \,.
}
Using these and the Klein-Gordon equation in the form \eqref{eq:rescaledKG-eq2} we find
\begin{align*}
 \tau^2 & \hdel(f_0) - 2 \tau^2 f_0 + m^2 \hdel (\phi^2)
\\
%&= 
%	\tau^2 \hdel \phi' \big( 2\phi ' + 3 N^{-1} \phi \hN\big)
%	+ 2 \tau^2 \nabla^a \hdel \phi \nabla_a \phi
%	-3 \tau^2 (\phi')^2 \hN + 3 \tau^2 N^{-1} \phi \phi' (\tfrac32 -N) 
%\\
%& \quad
%	+ \tau^2 \big(- 3\phi \phi' + \tfrac92 \phi^2 N^{-1} - \tfrac32 \phi^2 \big) \hdel N^{-1}
%	+\tau^2 (-2N \Sigma^{ab} - 2 \hN g^{ab} ) \nabla_a \phi \nabla_b \phi
%	- 2 \tau^2 f_0  + 2m^2 \phi \hdel \phi
%\\
%& = 
%	\tau^2 \big( 2\phi ' + 3 N^{-1} \phi \hN\big)\left( N \Delta \phi + \nabla^a N \nabla_a \phi -3 \hN \phi' +\phi' + \tfrac92 N^{-1} \phi \hN+ {\tfrac34} N^{-1} \phi + \tfrac32 \phi \hdel N^{-1} - \tau^{-2} m^2 N \phi \right)
%\\
%&\quad
%	+ 2 \tau^2 \nabla^a (\hdel \phi \nabla_a \phi) 
%	- 2 \tau^2 \hdel \phi \Delta \phi 
%	-3 \tau^2(\phi')^2 \hN 
%	+ 3 \tau^2 N^{-1} \phi \phi' (\tfrac32-N) 
%	+ \tau^2 \big(- 3\phi \phi' + \tfrac92 \phi^2 N^{-1} - \tfrac32 \phi^2 \big) \hdel N^{-1}
%\\
%&\quad
%	+ \tau^2 (-2N \Sigma^{ab} - 2 \hN g^{ab} ) \nabla_a \phi\nabla_b \phi
%	- 2 \tau^2\left( (\phi')^2 + \nabla^a \phi \nabla_a \phi + 3 N^{-1} \phi \phi' \hN + \tfrac32 N^{-2} \phi^2   (\tfrac32 -N )\right) + 2m^2 \phi N \phi'
%\\
& =  
	2 \tau^2 \nabla^a (\hdel \phi \nabla_a \phi) 
	+ 3 \tau^2 \nabla^a (\phi \hN \nabla_a \phi) 
	+\tau^2 \nabla_a N  \big( 3 \hN N^{-1} \phi \nabla_a \phi +2 \phi' \nabla_a \phi - \phi \nabla_a \phi \Big)
\\
&\quad
	+ \tau^2 \hN \Big( -3 (\phi')^2 - 5 (\nabla \phi)^2 + 6 N^{-1} \phi \phi' - 9 \hN (N^{-1} \phi) \phi' + \tfrac92 (N-\tfrac32) (N^{-1} \phi)^2 \Big)
\\
& \quad
	- \tau^2 \Big( 
	3  N^{-2} \phi^2 (\tfrac32 - N)
	+2 (\nabla \phi )^2 \Big)
	+ 3 \tau^2 N^{-1} \phi \phi' (2-N)
%\\
%&\quad
%	+ \tau^2 \hdel N^{-1} \left( 3 \phi \phi' + \tfrac92 N^{-1} \phi^2 \hN - 3 \phi \phi' + \tfrac92 \phi^2 N^{-1} - \tfrac32 \phi^2  \right)  
%\\
%& \quad
	- 2N \tau^2 \Sigma^{ab} \nabla_a \phi \nabla_b \phi
	-3 m^2 \hN \phi^2 .
\end{align*}
Note the terms involving $\hdel N$ cancelled. Combining this with \eqref{eq:pT-Ephi} we find
\begin{align*}
\p_T \tEphi{0}
& \lesssim
	 \big( \| \Sigma \|_{H^2} + \| \hN \|_{H^3} \big) \tEphi{0}
	+  \tau^2 \int \phi^2 \sqrt{g} +  \int ||\tau|^{1/2} \phi |\tau|^{3/2} \phi'|\sqrt{g}
\\
& \lesssim 
		 \big( \| \Sigma \|_{H^2} + \| \hN \|_{H^3}
		 + | \tau | \big) \tEphi{0}.
\end{align*}
Applying Gr\"onwall's inequality yields the result. 
\end{proof}

%%%%%%%%%%%%%%%%%%%%%%%%%%%%%%%%%%%%%%%%%%%%%%%%%%%%%
\subsection{Higher order modified Klein-Gordon energies}
Now we calculate the time derivatives of higher-order energy norms for the Klein-Gordon field. Recall the definition \eqref{eq:modified-energy-KG} of the modified $L^2$-energy, which we repeat again here for convenience. 
\begin{align*}
\tilde E_k(\phi) 
 &:=
	\int_\Sigma \tau^2 (-1)^k \Big[ 
	\phi' \Delta^k \phi' - \phi \Delta^{k+1} \phi 
+ 3 \phi'\Delta^k \big( N^{-1} \phi \hN \big) 
	+ \tfrac32 N^{-1}\phi \Delta^k \left( (\tfrac32-N) N^{-1} \phi \right) 
	\Big] \sqrt{g}
\\
& \qquad
	+ \int_\Sigma m^2 (-1)^k \phi \Delta^k \phi \sqrt{g},
\\
\tEphi{\ell} &:= \sum_{k=0}^\ell \tilde E_k(\phi).
\end{align*}

The main energy estimate for the modified higher order energies is given in the following proposition.
\begin{prop} \label{prop:prelim-est-Ek-KG}
Assume the bootstrap assumptions \eqref{eq:bootstraps}  hold, then the higher order energies for $1 \leq \ell \leq N$ satisfy
\eq{
\begin{split}
\p_T \tEphi{\ell}
& \lesssim
	\Big(
	\| \hN \|_{H^3} 
	+  \|\Sigma \|_{H^3} 
	+ \| \Sigma \|_{H^\ell}
	+ \| \hN \|_{H^{\ell+1}}
	+ |\tau|
	\Big)\tEphi{\ell} 
\\
& \quad
	+ \Big( \| \hN \|_{H^{\ell+1}} + \| \Sigma \|_{H^\ell} \Big) \tEphi{3}
	+ |\tau| \varepsilon^4
	+ \sum_{k=1}^\ell \Big| \int_\Sigma B_k \sqrt{g} \Big| ,
\end{split}
}
where $B_k$ denote the border-line terms, which for $k\geq 1$ are defined by
$$ 
B_k:= m^2 \phi' [N, \Delta^k] \phi .
$$

\end{prop}

\begin{proof}
In a similar way to Proposition \ref{prop:Ezero-KG}, we have 
\begin{align*}
\p_T \tilde E_k(\phi) 
\leq 
	\| \hN \|_{L^\infty} \tilde E_k(\phi) 
	+ \Big| \int_\Sigma  (-1)^k \big( \tau^2 \hdel (f_k) - 2 \tau^2 f_k +  m^2 \hdel ( \phi \Delta^k \phi ) \big) \sqrt{g} \Big| 
\end{align*}
where $f_k$ is the integrand inside the square brackets of $\tilde E_k(\phi)$ above. 
Hence we calculate the final term above and use the Klein-Gordon equation \eqref{eq:rescaledKG-eq2} to simplify.
For some function $u$ we have, by repeated applications of \eqref{eq:hdelg-up},
\begin{align*}
\hdel(\Delta^k u) 
&=	
	\hdel \big( g^{a_1 a_2} \cdots g^{a_{2k-1}a_{2k}} \nabla_{a_1} \nabla_{a_2} \cdots \nabla_{a_{2k-1}} \nabla_{a_{2k}} u  \big)
	\\
%&=
%	(\hdel g^{a_1 a_2} ) \nabla_{a_1} \nabla_{a_2} \Delta \cdots \Delta u
%	+ \ldots (\hdel g^{a_{2k-1}a_{2k}} ) \Delta \cdots\Delta \nabla_{a_{2k-1}} \nabla_{a_{2k}} u  
%\\
%& \quad
%	+ g^{a_1 a_2} \cdots g^{a_{2k-1}a_{2k}} \hdel \big( \nabla_{a_1} \nabla_{a_2} \cdots \nabla_{a_{2k-1}} \nabla_{a_{2k}} u \big)
%\\
& =	
	\Delta^k (\hdel u) 
	-2N \sum_{i=1}^k g^{a_1 a_2} \cdots \Sigma^{a_{2i-1} a_{2i}} \cdots g^{a_{2k} a_{2k-1}}\nabla_{a_1} \nabla_{a_2} \cdots \nabla_{a_{2i}} \nabla_{a_{2i-1}} \cdots \nabla_{a_{2k-1}} \nabla_{a_{2k}} u
\\
& \quad
	 -2\hN |k| \Delta^k u
	+  g^{a_1 a_2} \cdots g^{a_{2k-1}a_{2k}}  [\hdel,  \nabla_{a_1} \nabla_{a_2} \cdots \nabla_{a_{2k-1}} \nabla_{a_{2k}} ] u \,.
\end{align*}
We introduce the following compact notation for the second and last terms. 
\eq{ \begin{aligned}
\Sigma^I \Delta_I^k u &:= 	\sum_{i=1}^k g^{a_1 a_2} \cdots \Sigma^{a_{2i-1} a_{2i}} \cdots g^{a_{2k} a_{2k-1}}\nabla_{a_1} \nabla_{a_2} \cdots \nabla_{a_{2i}} \nabla_{a_{2i-1}} \cdots \nabla_{a_{2k-1}} \nabla_{a_{2k}} u
\\
[\hdel, \Delta^k]u &:= g^{a_1 a_2} \cdots g^{a_{2k-1}a_{2k}}  [\hdel,  \nabla_{a_1} \nabla_{a_2} \cdots \nabla_{a_{2k-1}} \nabla_{a_{2k}} ] u \,.
\end{aligned}
}
Thus
$$
\Big| \int_\Sigma (-1)^k  \big( \tau^2 \hdel (f_k) - 2 \tau^2 f_k + m^2 \hdel ( \phi \Delta^k \phi )\big) \sqrt{g} \Big|
\lesssim 
	I_k^1 + I_k^2 + I^3_k + C^1_k + C^2_k+ \Big| \int_\Sigma B_k \sqrt{g} \Big|
$$ 
where we define the lower-order integrals by
\eq{ \label{eq:Ik-def} \begin{aligned}
I_k^1
&:=
	\Big| \int_\Sigma \tau^2 \hN \big( |k| f_k - \phi \Delta^{k+1} \phi \big) \sqrt{g}\Big|
	+ \Big| \int_\Sigma \tau^2 N \phi' \Sigma^I \Delta^k_I \phi' \sqrt{g} \Big| 
	+  \Big| \int_\Sigma \tau^2 N \phi \Sigma^I \Delta^{k+1}_I \phi \sqrt{g} \Big| 
\\
&  \quad
	+ \Big| \int_\Sigma \tau^2 N \phi' \Sigma^I \Delta^k_I (N^{-1} \phi \hN) \sqrt{g} \Big| 
	+ \Big| \int_\Sigma \tau^2 \phi \Sigma^I \Delta^k_I \big( (\tfrac32-N) N^{-1} \phi \big) \sqrt{g} \Big| 
\\
& \quad
	  +  \Big| \int_\Sigma \hN m^2 \phi \Delta^k \phi \sqrt{g} \Big|
	  + \Big| \int_\Sigma m^2 \phi \Sigma^I \Delta^k_I \phi \sqrt{g} \Big| 
\\
I_k^2
&:= 	\Big| \int_\Sigma  \tau^2 
	 N \Delta \phi \Delta^k (N^{-1} \phi \hN) \sqrt{g}\Big|
	+ \Big|
	\int_\Sigma \tau^2 
	\nabla^a N \nabla_ a\phi \Delta^k \phi'  \sqrt{g}\Big|
	+\Big|
	\int_\Sigma \tau^2 
	\nabla^a N  \nabla_a \phi \Delta^k(N^{-1} \phi \hN) \sqrt{g} \Big|
\\
I_k^3
& :=
	\Big| \int_\Sigma \tau^2 
	 \Big(
	-3 \hN \phi' \Delta^k \phi'
	+ (15-3N) \phi' \Delta^k (\hN N^{-1} \phi)
	+ \tfrac92 (N-\tfrac32) N^{-1} \phi \Delta^k(N^{-1} \phi \hN) \Big) \sqrt{g} \Big|
\\
& \quad
	+ \Big| \int_\Sigma \tau^2 		
	(2-N) N^{-1} \phi \Delta^k \phi' \sqrt{g}\Big|
	+\Big| 
	\int_\Sigma \tau^2 (\tfrac32 -N )\phi \Delta^k (N^{-1} \phi)
	\sqrt{g}\Big|
\end{aligned} }
and the integrals involving commutators by
\eq{ \label{eq:Ck-def} \begin{aligned}
C_k^1 &:=
		 \Big| \int_\Sigma m^2 \phi[\hdel, \Delta^k] \phi  \sqrt{g} \Big|
		+ \Big| \int_\Sigma \tau^2 \Big( \phi' [\hdel,\Delta^k ] \phi' 
	-  \phi [\hdel, \Delta^{k+1}] \phi	
	+ 3 \phi' [\hdel,\Delta^k] \big( N^{-1} \phi \hN\big)
\\
& \quad
	+ \tfrac32 (N^{-1} \phi) [\hdel, \Delta^k] \big( (\tfrac32-N) N^{-1} \phi \big) 	\Big)\sqrt{g} \Big|
	\\
C_k^2 &:=
	\Big| \int_\Sigma \tau^2 \Big(
	- 2 \Delta \phi [\Delta^k, N] \phi' 
	-\tfrac32 N^{-1} \phi [\Delta^k, N] \phi'
	\Big) \sqrt{g} \Big|.
\end{aligned} }
Finally the terms without decaying factors (for example, without factors of $|\tau|$ or $\| \hN \|_{L^\infty}$), and which will require additional care to control, are the following. 
\begin{align*}
B_k &:= 
	 m^2 \phi' [N, \Delta^k] \phi.
\end{align*}
Note the terms involving $\hdel N$ have cancelled with each other. The terms $I_k$ will be controlled using Lemmas \ref{lem:Lk1} and \ref{lem:Lk23}. The commutator terms $C_k$ will be controlled using Lemmas \ref{lem:commutator-hdel} and \ref{lem:commutatorsN} below. Summing these estimates for $k=0$ to $k=\ell$, noting that $k=0$ is covered using Proposition \ref{prop:Ezero-KG}, yields the claim. 
\end{proof}

%%%%%%%%%%%%%%%%%%%%%%%%%%%%%%%%%%%%%%%%%%%%%%%%%%%%%
\subsection{Auxiliary lemmas}
As mentioned in the foregoing proof we require a series of lemmas that are used in the proof of the main energy estimate above. We list and prove those in the following. The main strategy throughout this section is to integrate by parts on each term and distribute $k \geq 1$ derivatives while also making use of the Sobolev embeddings $H^2 \hookrightarrow L^\infty$ and $H^1 \hookrightarrow L^4$.

\begin{lem} \label{lem:common-Hk-est}
For $k \geq 1$ and general functions $v, u, $ and $w$ we have
\begin{align*}
&\Big| \int_\Sigma vu \Delta^k w \sqrt{g} \Big|
	\lesssim
	\| v \|_{L^\infty} 
	\|  u \|_{H^{k}} \|  w \|_{H^{k}}
	+ \| u|_{L^\infty} 
	\| v \|_{H^{k}} \|  w \|_{H^{k}}
	+ \| v\|_{H^k}
	\|  u \|_{H^k}
	\| w \|_{H^k}.
\end{align*}
The sums involving $|I| , |J| \leq k-1$ do not appear if $k=1$. 
Also assuming the bootstrap assumptions \eqref{eq:bootstraps}  hold, then for $k \leq N$
\begin{align} \label{eq:est-Ninv-phi}
\| \tau N^{-1} \phi \|_{H^k}
&\lesssim
	\big( |\tau| + \| \hN \|_{H^{k+1}}\big) \tEphi{k}^{1/2}
\\
\| \tau (3-N) N^{-1} \phi  \|_{H^k} \label{eq:est-NhN}
& \lesssim
	|\tau| \tEphi{k}^{1/2} 
	+ | \tau | \varepsilon^2
\\
\| \tau (\alpha -N) N^{-1} \phi \|_{H^k} \label{eq:est-NalphaN}
& \lesssim
		\big( |\tau| + \| \hN \|_{H^{k+1}}\big) \tEphi{k}^{1/2}
		+ \| \hN \|_{H^k}\tEphi{1}^{1/2} 
\end{align}
where $\alpha \neq 3$ .
\end{lem}

\begin{proof}
Using the Sobolev embeddings $H^2 \hookrightarrow L^\infty$ and $H^1 \hookrightarrow L^4$ and integration by parts on general functions $v, u, $ and $w$ gives, for $k \geq 1$,
\begin{align*}
&\Big| \int_\Sigma vu \Delta^k w \sqrt{g} \Big|
	\lesssim
	\int_\Sigma \Big| (\nabla)^\oddk \Delta^\fk (v u)
	(\nabla)^\oddk  \Delta^\fk (w) \Big| \sqrt{g}	
\\
& \quad \lesssim
	\| v \|_{L^\infty} 
	\|  u \|_{H^{k}} \|  w \|_{H^{k}}
	+ \| u\|_{L^\infty} 
	\| v \|_{H^{k}} \|  w \|_{H^{k}}
	+ \sum_{1 \leq |I| \leq k-1} \|  \nabla^I v\|_{L^4} 
	\sum_{1 \leq |J|\leq k-1} \| \nabla^{J} \phi\|_{L^4}
	\|  w \|_{H^{k}}
\\
& \quad \lesssim
	\| v \|_{L^\infty} 
	\|  u \|_{H^{k}} \|  w \|_{H^{k}}
	+ \| u|_{L^\infty} 
	\| v \|_{H^{k}} \|  w \|_{H^{k}}
	+ \| v\|_{H^k}
	\|  u \|_{H^k}
	\| w \|_{H^k}.
\end{align*}
In general sums involving $|I| , |J| \leq k-1$ do not appear if $k=1$. 
We also have
\begin{align*}
\| \tau N^{-1} \phi \hN \|_{H^k}
& \lesssim
	\| \tau \phi \|_{H^k} + \| \tau N^{-1} \phi \|_{H^k}
\\
& \lesssim 
	\frac{|\tau|}{m^2} \tEphi{k}^{1/2} 
	+ |\tau| \| N^{-1} \| _{L^\infty} \tEphi{k}^{1/2} 
	+ \| \tau \phi \|_{L^\infty} \| \hN \|_{H^k} 
	+ |\tau| \| \hN \|_{H^k} \tEphi{k}^{1/2} 
\\
& \lesssim 
	|\tau| \tEphi{k}^{1/2} 
	+ | \tau | \varepsilon^2.
\end{align*}
and also
\begin{align*}
\| \tau N^{-1} \phi \|_{H^k}
&=
	|\tau| \| N^{-1} \|_{L^\infty} \| \phi\|_{H^k} 
	+ \Big( \sum_{|I|+|J| +1\leq k} \int_\Sigma \tau^2 | \nabla^{I+1} N^{-1}|^2| \nabla^{J} \phi|^2 \sqrt{g} \Big)^{1/2} 
\\
& \lesssim 
	|\tau| \tEphi{k}^{1/2} 
	+ \sum_{1 \leq |I|\leq k} \| \nabla^I N^{-1} \|_{L^4}
	\sum_{|J|\leq k-1} \| \tau \nabla^J \phi \|_{L^4}
\\
&\lesssim
	\big( |\tau| + \| \hN \|_{H^{k+1}}\big) \tEphi{k}^{1/2}. 
\end{align*}

Finally although $\hN$ is small, there are several terms involving $N- \alpha$ where $\alpha \neq 3$. In this case we take care to extract $N-\alpha$ in $L^\infty$ when no derivatives hit it, but when derivatives do hit this term we note that $\nabla (N-\alpha) = \nabla \hN$ and this is small. 
\begin{align*}
\| \tau (\alpha - N) N^{-1} \phi \|_{H^k}
& \lesssim
	\| \alpha -N \|_{L^\infty} \| \tau N^{-1} \phi \|_{H^k}
	+ \| \tau N^{-1} \phi \|_{L^\infty} \| \hN \|_{H^k}
	+ \| \tau N^{-1} \phi \|_{H^k} \| \hN\|_{H^k}
\\
& \lesssim
	\| \tau N^{-1} \phi \|_{H^k}
	+ \| N^{-1} \|_{L^\infty} \| \tau \phi \|_{H^2} \| \hN \|_{H^k}
\\
& \lesssim
		\big( |\tau| + \| \hN \|_{H^{k+1}}\big) \tEphi{k}^{1/2}
		+ \| \hN \|_{H^k}\tEphi{1}^{1/2} .
\end{align*}
\end{proof}

The first set of lower-order integrals $I_k^1$ are controlled in the following Lemma. 
\begin{lem} \label{lem:Lk1}
\eq{ I_k^1 \lesssim 	
	\big( \| \hN \|_{H^2} + \| \Sigma \|_{H^2} \big) \tEphi{k} 
	+ \big( \| \hN \|_{H^k} + \| \Sigma \|_{H^k} \big) \tEphi{2}
	+ \big( \| \hN \|_{H^k}+ \| \Sigma \|_{H^k} \big)  \tEphi{k}.
}
\end{lem}

\begin{proof}
The results of Lemma \ref{lem:common-Hk-est} allow us to easily obtain
\begin{align*}
\Big| & \int_\Sigma \hN m^2 \phi \Delta^k \phi \sqrt{g} \Big| + \Big| \int_\Sigma \tau^2 \hN \big( |k| f_k - \phi \Delta^{k+1} \phi \big) \sqrt{g}\Big|
\\
& \quad
\lesssim
	\| \hN \|_{H^2} \tEphi{k} 
	+ \| \hN \|_{H^k} \tEphi{2}
	+ \| \hN \|_{H^k} \tEphi{k}.
\end{align*}
The remaining terms involving contractions with $\Sigma^{ab}$ can similarly be estimated. For example
\begin{align*}
\Big| \int_\Sigma m^2 \phi \Sigma^I \Delta^k_I \phi \sqrt{g} \Big| 
& \lesssim
	\Big( \| \Sigma \|_{H^2} \| m \phi \|_{H^k} + \| m \phi \|_{L^\infty} \| \Sigma \|_{H^k} + \| m \phi \|_{H^k} \| \Sigma \|_{H^k} \Big) \| m \phi \|_{H^k}
\\
& \lesssim
		\| \Sigma \|_{H^2} \tEphi{k} 
	+ \| \Sigma \|_{H^k} \tEphi{2}
	+ \| \Sigma \|_{H^k} \tEphi{k}.
\end{align*}
\end{proof}

\begin{lem} \label{lem:Lk23}
Assume the bootstrap assumptions \eqref{eq:bootstraps}  hold, then the following estimate holds.
\begin{align*}
I_k^2+I_k^3
& \lesssim
	\big( \|\hN\|_{H^{k+1}} 
	+ \| \hN \|_{H^3} 
	\big)\tEphi{k}
	+ \| \hN \|_{H^{k+1}} \tEphi{3}
	+ |\tau| \varepsilon^4
\end{align*}
\end{lem}

\begin{proof}
 We make frequent use of the identities from Lemma \ref{lem:common-Hk-est} and also the identity \eqref{eq:nablaN-inverse} for derivatives of $N^{-1}$. 
For the first term in $L^2_k$, we integrate by parts only $k-1$ times so that we avoid terms such as $\| \tau \phi \|_{H^{k+2}}$ since this is only controlled by $\tEphi{k+1}$. 
\begin{align*}
&\left| \int_\Sigma (-1)^k \tau^2 N \Delta \phi \Delta^k ( N^{-1} \phi \hN ) \sqrt{g} \right|
\lesssim
	\| \tau N \Delta \phi \|_{H^{k-1}} \| \tau (N^{-1} \hN) \phi \|_{H^{k+1}}
\\
& \quad \lesssim
	\Big( \| N \|_{L^\infty} \| \tau \phi \|_{H^{k+1}}
	+ \| \tau \Delta \phi \|_{L^\infty} \sum_{1 \leq |I|\leq k-1} \| \nabla^I N\|_{L^2}
	+ \sum_{1 \leq |I|\leq k-2} \| \nabla^I N\|_{L^4}
	\sum_{1 \leq |J|\leq k-2} \| \tau \nabla^J \Delta \phi \|_{L^4}
	\Big)
\\
& \quad \quad \times
	\Big( 
		\| N^{-1} \hN \|_{L^\infty} \| \tau \phi \|_{H^{k+1}}
		+ |\tau| \| \phi \|_{L^\infty} \sum_{1 \leq |I|\leq k+1} \| \nabla^I N^{-1} \|_{L^2}
		+ \sum_{1 \leq |I|\leq k} \| \nabla^I N^{-1} \|_{L^4}
	\sum_{1 \leq |J|\leq k} \| \tau \nabla^J  \phi \|_{L^4}
	\Big)
\\
& \quad \lesssim
	\Big( \| \tau \phi \|_{H^{k+1}} 
	+ \| \tau \phi \|_{H^4} \| \hN \|_{H^{k-1}} \Big)
	\Big(
		|\tau| \| \phi \|_{H^2} \| \hN \|_{H^{k+1}}
		+ \| \hN \|_{H^{k+1}} \| \tau \phi \|_{H^{k+1}}
		\Big)
\\
& \quad \lesssim
	\| \hN \|_{H^{k+1}} \tEphi{3}
	+ \| \hN \|_{H^{k+1}} \Ephi{k}	.
\end{align*}
Sums involving $|I| \leq k-1$ or $|I| \leq k-2$ do not exist for $k =1$ and $k=2$ respectively.  The key point above is the estimate 
$$
\| \tau N^{-1} \hN \phi \|_{H^{k+1}} 
\lesssim 
\| \tau \phi \|_{H^{k+1}}\| \hN \|_{H^{k+1}} + |\tau| \varepsilon^2
\lesssim
\tEphi{k}^{1/2} \| \hN \|_{H^{k+1}} + | \tau | \varepsilon^2.
$$ 
Note the first term $\tEphi{k}^{1/2} \| \hN \|_{H^{k+1}}$ above is worse than the term $\tEphi{k}^{1/2} |\tau|$ from \eqref{eq:est-NhN}. This is because have more derivatives to distribute and so we must allow for a term with both high derivatives in $\hN$ and $\phi$. 

For the remaining terms of $L_k^1$ we integrate by parts $k$ times to obtain
\begin{align*}
&\Big|\int_\Sigma \tau^2 
	\nabla^a N \nabla_ a\phi \Delta^k \phi'  \Big|
	+\Big|
	\int_\Sigma \tau^2 
	\nabla^a N  \nabla_a \phi \Delta^k(N^{-1} \phi \hN) \Big|
\\
& \lesssim \Big(
	\| \nabla N \|_{L^\infty} \| \tau \nabla \phi \|_{H^k} 
	+ \| \tau \nabla \phi \| _{L^\infty} \| \nabla N \|_{H^k}
	+ \| \nabla N \|_{H^k} \| \tau \nabla \phi \|_{H^k}
	\Big)
	\big( \| \tau \phi' \|_{H^k} 
	+ \| \tau N^{-1} \phi \hN \|_{H^k}
	\big)
\\
& \lesssim	
	\big( 
	\| \hN \|_{H^3} + \| \hN \|_{H^{k+1}} 
	\big) \Ephi{k}
	+ |\tau| \| \hN \|_{H^{k+1}} \tEphi{3}
	+ \varepsilon^4 |\tau| \| \hN \|_{H^2}.
\end{align*}
Thus
\begin{align*}
I_k^2
& \lesssim
	\big( \|\hN\|_{H^{k+1}} 
	+ \| \hN \|_{H^3} 
	\big)\tEphi{k}
	+ \| \hN \|_{H^{k+1}} \tEphi{3}
	+ |\tau| \varepsilon^4.
\end{align*}
Turning to $I^3_k$, we see that the last two terms contain no factors of $\hN$. We estimate these using \eqref{eq:est-NalphaN} to obtain
\eq{ \label{eq:est-2-N} \begin{aligned}
\Big| \int_\Sigma \tau^2 		
	(2-N) N^{-1} \phi \Delta^k \phi' \sqrt{g}\Big|
& \lesssim
\| \tau (2-N)N^{-1} \phi  \|_{H^k} \| \tau \phi' \|_{H^k}
\\
& \lesssim
	\big( |\tau| + \| \hN \|_{H^{k+1}}\big) \tEphi{k}
	+ \| \hN \|_{H^k} \tEphi{2}. 
\end{aligned}}
The other terms of $I^3_k$ are estimated in a similar manner, using instead \eqref{eq:est-NhN}.  Thus
\begin{align*}
I^3_k
& \lesssim
	\big( \|\hN\|_{H^{k+1}} 
	+ |\tau| 
	+ \| \hN \|_{H^2} 
	\big)\tEphi{k}
	+ \| \hN \|_{H^k} \tEphi{2}
	+ |\tau| \varepsilon^4.
\end{align*}

\end{proof}

We now estimate the commutator terms $C_k$, divided into those commutators of the form $[\hdel, \Delta^k]$, see Lemma \ref{lem:commutator-hdel}, or those of the form $[\Delta^k, N]$, see Lemma \ref{lem:commutatorsN}. We start with the following identity, adapted from \cite[(6.6)]{AF17} and \cite{ChMo01}. 

\begin{lem}[Commutator identity] \label{lem:commutator-hdel}
For some general functions $v,w$ and $k \geq 1$ we have
\eq{
\begin{split}
&\Big| \int_\Sigma v [\hdel, \Delta^k] w \sqrt{g} \Big|
\\
&\lesssim
 \big( 
		\| v \|_{H^k}  \| w \|_{H^{k-1}} 
		+ 
		\| w \|_{H^k}  \| v \|_{H^{k-1}} 
		\big)
		\big( 
	\| \Sigma \|_{H^3}
	+ \| \hN \|_{H^3}
	+ \| \Sigma \|_{H^k} + \| \hN \|_{H^k}
\big).
\end{split}
}
\end{lem}

\begin{proof}
From \cite{AF17} we have
\begin{align*}
& [ \hdel, \Delta^k] w 
:= 
	g^{a_1 a_2} \cdots g^{a_{2k-1}} g^{a_{2k}} [\hdel, \nabla_{a_1} \nabla_{a_2} \cdots \nabla_{a_{2k-1}} \nabla_{a_{2k}} ] w
\\
&=
	g^{a_1 a_2} \cdots g^{a_{2k-1}} g^{a_{2k}} 
	\Big[ \sum_{i \leq 2k-1} \sum_{i+1 \leq j \leq 2k}
	\nabla_{a_1} \cdots \nabla_{a_{i-1}}
	\Big( \nabla_{a_{i+1}} \cdots \nabla_{a_{j-1}} \nabla_c \nabla_{a_{j+1}} \cdots \nabla_{a_{2k}} (w) \cdot K^c_{a_j a_i} \Big)
	\Big]
\end{align*}
where 
$$
K_{bc}^a := \nabla_b(Nk^a_c) + \nabla_c (N k^a_b) - \nabla^a (k_{cb}) .
$$
Thus after integration by parts, we find
\begin{align*}
\Big| &\int_\Sigma v [\hdel, \Delta^k] w \sqrt{g} \Big|
%& =
%	\sum_{2 \leq i \leq 2k-1} \sum_{i+1 \leq j \leq 2k}
%	\int_\Sigma 
%	v g^{a_1 a_2} \cdots g^{a_{2k-1}} g^{a_{2k}}  \nabla_{a_1} \cdots \nabla_{a_{i-1}} \left(
%	\nabla_{a_{i+1}} \cdots \nabla_{a_{j-1}} \nabla_c \nabla_{a_{j+1}} \cdots \nabla_{a_{2k}} (w) \cdot K^c_{a_j a_i} \right) \sqrt{g}
=
	\Big| \sum_{\substack{|I|+|J|=2k-1 \\ |J| \geq 1}} 
	 \int_\Sigma c_{IJ} v \nabla^I(\nabla^J(w)K)  \sqrt{g} \Big|
\\
&
	\lesssim
	\Big( \sum_{\substack{|I|+|J|=2k-1 \\  1 \leq |J| \leq k}}
	+
	\sum_{\substack{|I|+|J|=2k-1 \\ |I| \leq k}} 
	\Big) c_{IJ} \int_\Sigma  v \nabla^I(\nabla^J(w)K) \sqrt{g}
\\
& \lesssim
	\sum_{\substack{|I''|+|J|=k-1 \\ |I'| = k \\ |J| \geq 1}}
	\int_\Sigma \big| \nabla^{I'} v \nabla^{I''} \left( \nabla^J w \cdot K \right) \big| \sqrt{g}
	+
	\sum_{\substack{|I| +|J''|=k-1 \\ |J'| = k \\ |J''| \geq 1}}
	\int_\Sigma \big| \nabla^{J'} w \nabla^{J''} \left( \nabla^I v \cdot K \right) \big| \sqrt{g}
%\\
%& \lesssim		
%	\| v \|_{H^k} \left( 
%	\|K\|_{L^\infty} \| w \|_{H^{k-1}} 
%	+ \| w \|_{H^{k-1}} \| K \|_{H^{k-1}}
%	\right)
%	+ \| w \|_{H^k} \left( 
%	\|K\|_{L^\infty} \| v \|_{H^{k-1}} 
%	+ \| v \|_{H^{k-1}} \| K \|_{H^{k-1}}
%	\right)
\\
& \lesssim
		\left( 
		\| v \|_{H^k}  \| w \|_{H^{k-1}} 
		+ 
		\| w \|_{H^k}  \| v \|_{H^{k-1}} 
		\right)
		\left( 
	\|K\|_{L^\infty} + \| K \|_{H^{k-1}}
	\right).
\end{align*}
Finally recall $k = \Sigma + g/3$ so that $K = \nabla(Nk) = \nabla(N\Sigma) + {g \over 3} \nabla N$. Thus
\begin{align*}
&\|K\|_{L^\infty} + \| K \|_{H^{k-1}}
%\\
%& \quad \lesssim
%	\| N \|_{L^\infty} \| \nabla \Sigma \|_{L^\infty} 
%	+ \| \Sigma \|_{L^\infty} \| \nabla N \|_{L^\infty}
%	+ \| \nabla N \|_{L^\infty}
%	+ \| N \nabla \Sigma \|_{H^{k-1}}
%	+ \| \Sigma \nabla N \|_{H^{k-1}}
%	+ \| \nabla N \|_{H^{k-1}}
%\\
\lesssim
	\| \Sigma \|_{H^3}
	+ \| \hN \|_{H^3}
	+ \| \Sigma \|_{H^k} + \| \hN \|_{H^k}.
\end{align*}
\end{proof}

\begin{lem}
Assume the bootstrap assumptions \eqref{eq:bootstraps}  hold, then
\eq{C_k^1
\lesssim
	\big( \tEphi{k}
	+ | \tau | \varepsilon^4 \big)
	\big( 
	\|K\|_{L^\infty} + \| K \|_{H^{k-1}}
	\big).
}
\end{lem}

\begin{proof}
Using Lemma \ref{lem:commutator-hdel} the `symmetric' terms are controlled by
\begin{align*}
\Big| \int_\Sigma & 
	\big( \tau^2\phi' [\hdel,\Delta^k ] \phi'  
	+  m^2 
	\phi[\hdel, \Delta^k] \phi 
	\big) \sqrt{g} \Big|
%\\
%&\lesssim
%	\| K \|_{L^\infty} \Ephi{k}^{1/2} \Ephi{k-1}^{1/2} 
%	+ \Ephi{k}^{1/2} \Ephi{k-1}^{1/2} \|K\|_{H^{k-1}} 
%\\
%& 
\lesssim
	\left( 
	\|K\|_{L^\infty} + \| K \|_{H^{k-1}}
	\right) \tEphi{k}.
\end{align*}
The remaining terms are
\begin{align*}
\Big| \int_\Sigma & \tau^2\phi [\hdel, \Delta^{k+1}] \phi	\sqrt{g} \Big|
	+ \Big| \int_\Sigma 
		\tau^2 \phi' [\hdel,\Delta^k] \big( N^{-1} \phi \hN \big)\sqrt{g} \Big|
	+\Big| \int_\Sigma \tau^2 (N^{-1} \phi)  [\hdel, \Delta^k] \big( (\tfrac32-N) N^{-1} \phi \big) \sqrt{g} \Big|
\\
& \lesssim
	\Big( |\tau| \| \tau \phi \|_{H^{k+1}} \| \phi \|_{H^k}
	+ \| \tau \phi' \|_{H^k} \| \tau N^{-1} \phi \hN \|_{H^k}  
	+ \| \tau N^{-1} \phi \|_{H^k} \| \tau \phi(1-3/2N) \|_{H^k} 
	\Big)
	\big( 
	\|K\|_{L^\infty} + \| K \|_{H^{k-1}}
	\big)
\\
& \lesssim
	\Big( 
	|\tau| \Ephi{k}
	+ | \tau | \varepsilon^4
	+ \|\hN\|_{H^{k+1}} \Ephi{k} \Big)
	\big( 
	\|K\|_{L^\infty} + \| K \|_{H^{k-1}}
	\big).
\end{align*}
In the final line we used Lemma \ref{lem:common-Hk-est}. 

\end{proof}

The final result in this section controls commutators involving $N$. 
\begin{lem} \label{lem:commutatorsN}
Assume the bootstrap assumptions \eqref{eq:bootstraps}  hold, then
\eq{
C_k^2
\lesssim
	\| \hN \|_{H^3} \tEphi{k} 
	+ \| \hN \|_{H^{k+1}} \tEphi{3}
	+ \| \hN\|_{H^{k+1}} \tEphi{k} .
}
\end{lem}

\begin{proof}
First note the expansion
\begin{align*}
[\Delta^k, N] w 
&=
	\sum_{|I|+|J| = 2k-1} c_{IJ} \nabla^{I+1} N \nabla^J w
\end{align*}
where the constants $c_{IJ}$ are functions of $g$ and so will be bounded below by some large overall constant. For general functions $w,v$ and $k \geq 1$ we have the following 
\begin{align*}
&\int_\Sigma 
	v [\Delta^k, N] w \sqrt{g}
\\
&\lesssim
	\sum_{\substack{|I|+|J'|=k-1 \\ |J''|=k}} 
	\| \nabla^{J'}  v \nabla^{I+1} N \nabla^{J''} w \|_{L^1}
	+
	\sum_{\substack{|I'|+|J|=k-1 \\ |I''|=k+1}} 
	\|\nabla^{I'}  v \nabla^J w \nabla^{I''} N \|_{L^1}
\\
& \lesssim
	\Big( \|  v \|_{L^\infty} \| \hN \|_{H^k} 
	+ \| \nabla N \|_{L^\infty} \|  v \|_{H^{k-1}} 
	+ \sum_{|\alpha| \leq k-2} \|  \nabla^\alpha v \|_{L^4} \sum_{|\beta| \leq k-1} \| \hN \|_{L^4}
	\Big) \|\tau w \|_{H^k}
\\
& \quad
	+ \Big( \|  w \|_{L^\infty} \|  v \|_{H^{k-1}} 
	+ \|  v \|_{L^\infty} \|  w \|_{H^{k-1}} 
	+ \sum_{|\alpha| \leq k-2} \|  \nabla^\alpha v \|_{L^4} \sum_{|\beta| \leq k-2} \|  w\|_{L^4}
	\Big) \| \hN \|_{H^{k+1}}
\\
& \lesssim
		\Big( \| v \|_{H^2} + \| v \|_{H^{k-1}} \Big) \| \hN \|_{H^{k+1}} \| w \|_{H^k} 
		+ \Big( \| w\|_{H^k}\| \hN \|_{H^3} + \| w \|_{H^2} \| \hN \|_{H^{k+1}}\Big) \| \hN \|_{H^{k-1}} \| w \|_{H^k} .
\end{align*}
The claim follows, recalling for example the estimate \eqref{eq:est-Ninv-phi}. 

\end{proof}

%%%%%%%%%%%%%%%%%%%%%%%%%%%%%%%%%%%%%%%%%%%%%%%%%%%%%

%%%%%%%%%%%%%%%%%%%%%%%%%%%%%%%%%%%%%%%%%%%%%%%%%%%%%
\section{Lapse and shift estimates}
We first recall the following elliptic estimates from \cite[Proposition 17]{AF17} for the lapse and shift. 
\begin{prop} \label{prop:AF17-N-X}
Under appropriate smallness conditions we have the pointwise estimate $N \in (0, 3]$ and the following estimates
\eq{
\begin{split}
\| \hN \|_{H^\ell} & \leq	
	C \big( \| \Sigma \|^2_{H^{\ell-2}} 
	+ |\tau| \| \eta \|_{H^{\ell-2}} \big),
\\
\| X \|_{H^\ell} & \leq 
	C \big( 	\| \Sigma \|_{H^{\ell-2}}^2
	+ \| g-\gamma\|_{H^{\ell-1}}^2
	+ |\tau| \| \eta \|_{H^{\ell-3}}
	+ \tau^2 \| N j \|_{H^{\ell-2}}
	\big).
\end{split}
}

\end{prop}
Applied to the present case this yields the following estimate for the lapse function.
\begin{lem}[Lapse estimate] \label{lem:lapse-est}
Assume the bootstrap assumptions \eqref{eq:bootstraps}  hold, then for $2 \leq \ell \leq N+1$ we have
\eq{
\| \hN \|_{H^\ell} 
\lesssim
	\| \Sigma \|^2_{H^{\ell-2}} + |\tau| \| \rho \|_{L^\infty} + |\tau| \tEphi{\ell-2}
 }
and  furthermore
 \eq{
\| \hN \|_{H^\ell} 
 \lesssim 
	\varepsilon^2 |\tau|
	+|\tau| \tEphi{\ell-2}
}

\end{lem}

\begin{proof}
For $k\geq 0$ and some function $w$ we have
\begin{align}
\| w^2 \|_{H^k}^2 
%&= 
%	\sum_{|I| \leq k} \int_M | \nabla^I (w^2) |^2 \sqrt{g}
%\notag \\
& \lesssim
	\| w \|_{L^\infty}^2 \sum_{|I| \leq k} \int_\Sigma | \nabla^I w |^2 \sqrt{g}
	+ \sum_{|I|+|J|+2\leq k} \int_\Sigma 
	\big( | \nabla^{I+1} w |^4 + |\nabla^{J+1} w|^4  \big) \sqrt{g}
\notag \\
& \lesssim
	\| w \|_{L^\infty}^2 \|  w \|_{H^\ell}^2  + \sum_{|I|+1\leq k-1} \| \nabla^{I+1} w \|_{L^4}^4
\notag \\
& \lesssim
	\| w \|_{L^\infty}^2 \| w \|_{H^k}^2  + \| w \|_{H^k}^4.
\label{eq:Hk-square-product}
\end{align}
Recalling the definition of $\eta$ from \eqref{eq:eta} and the estimate from Proposition \ref{prop:AF17-N-X}
%$\eta = -{\tfrac12} (mu)^2 + 2 \left( \tau \left(\tfrac32 N^{-1} \phi - \phi' \right) \right)^2 $ 
we see
\begin{align*}
\| \hN \|_{H^\ell} 
& \lesssim
	\| \Sigma \|^2_{H^{\ell-2}} 
	+ |\tau| \| \eta \|_{H^{\ell-2}}
\\
& \lesssim
	\| \Sigma \|^2_{H^{\ell-2}}
	+ |\tau| \big( \| m\phi \|_{L^\infty} + \|\tau N^{-1} \phi \|_{L^\infty} + \| \tau \phi' \|_{L^\infty} \big) 
	 \tEphi{\ell-2}^{1/2} + |\tau| \tEphi{k-2}
	\\
& \lesssim
	\| \Sigma \|^2_{H^{\ell-2}} 
	+ |\tau| \tEphi{\ell-2}
	+ |\tau| \| \rho \|_{L^\infty}.
 \end{align*}

\end{proof}

\begin{rem}
It is crucial in the final line of the previous proof (specifically for  $\ell=2$) to use the pointwise estimate from Proposition \ref{prop:modified-cont} instead of the standard Sobolev estimate invoking $\tEphi{2}$, since the latter would have created a $e^{\gamma T}$ growth preventing the envisioned bootstrap argument. 
\end{rem}

\begin{lem}[Shift estimate]\label{lem:shift-est}
Assume the bootstrap assumptions \eqref{eq:bootstraps}  hold, then for $3 \leq \ell \leq N+1$
\begin{align*}
\| X \|_{H^\ell}
& \lesssim
	\| \Sigma \|_{H^{\ell-2}}^2
	+ \| g-\gamma\|_{H^{\ell-1}}^2
	+| \tau| \| \rho \|_{L^\infty}
	+  |\tau | \tEphi{\ell-2} 
\end{align*}
and furthermore
 \eq{
\| X \|_{H^\ell}
 \lesssim 
	\varepsilon^2 |\tau|
	+|\tau|\tEphi{\ell-2}.
}
\end{lem}

\begin{proof}
Recall the estimate from Proposition \ref{prop:AF17-N-X}.  We may use the estimate for $\| \eta \|_{H^k}$ derived in Lemma \ref{lem:lapse-est}, and will also need an estimate for the rescaled matter current
\eq{
j^a = \tau \big(\tfrac32 N^{-1} \phi - \phi' \big) \nabla^a \phi .
}
For $k \geq 0$  using the standard Sobolev embeddings, we have
\begin{align*}
\tau \| j \|_{H^k}
& \lesssim
	\| \tau \nabla \phi \|_{L^\infty} \tEphi{k}^{1/2} 
	+ \big( \| \tau N^{-1}\phi\|_{L^\infty} + \| \tau \phi'\|_{L^\infty} \big)
	\tEphi{k}^{1/2}
	+ \tEphi{k}
\\
& \lesssim
	\| \rho \|_{L^\infty} + \tEphi{k} .
\end{align*}
So for $3 \leq \ell \leq N+1$ we have
\begin{align*}
\| X \|_{H^\ell}
& \lesssim
	\| \Sigma \|_{H^{\ell-2}}^2
	+ \| g-\gamma\|_{H^{\ell-1}}^2
	+ |\tau| \| \eta \|_{H^{\ell-3}}
	+ \tau^2 \| N j \|_{H^{\ell-2}}
\\
& \lesssim
	\| \Sigma \|_{H^{\ell-2}}^2
	+ \| g-\gamma\|_{H^{\ell-1}}^2
	+ |\tau| \tEphi{\ell-3}+ |\tau| \| \rho \|_{L^\infty}
\\
& \quad
	+ |\tau| \Big( \| N \|_{L^\infty} \| \tau j \|_{H^{\ell-2}} 
	+ \sum_{|I|\leq \ell-2} \| \nabla^I \hN\|_{L^\infty} \| \tau j \|_{L^2} 
	+ \| \hN\|_{H^{\ell-2}} \| \tau j \|_{H^{\ell-2}}
	\Big)
\\
& \lesssim
	\| \Sigma \|_{H^{\ell-2}}^2
	+ \| g-\gamma\|_{H^{\ell-1}}^2
	+ |\tau| \tEphi{\ell-3}
	+ |\tau| \big(  
	\| \rho \|_{L^\infty}
	+  \tEphi{\ell-2} 
	+ \| \hN \|_{H^\ell} \tEphi{2} 
	\big).
\end{align*}
\end{proof}

\section{Hierarchy between Lapse and Klein-Gordon field}
In the following Lemma we estimate the borderline terms for the Klein-Gordon energy. 
\begin{lem}[Borderline terms]\label{lem:BL-terms}
Assume the bootstrap assumptions \eqref{eq:bootstraps}  hold, then for $1 \leq \ell \leq N $ we have
\eq{\label{eq:BL-decay-est}
\begin{split}
\sum_{k=1}^\ell \Big| \int_\Sigma BL_k \sqrt{g} \Big|
& \lesssim
	|\tau|^{-1} \varepsilon \tEphi{\ell}^{1/2} \| \hN \|_{H^\ell} 
	+ |\tau|^{-1} \tEphi{\ell}
		\|\hN\|_{H^2} 
		+ |\tau|^{-1} \tEphi{\ell}^{1/2} \tEphi{\ell-1}^{1/2} \| \hN \|_{H^{\ell+1}}.
\end{split}
}
\end{lem}

\begin{proof}
\begin{align*}
 \Big| \int_M & 
	  BL_k
	\sqrt{g} \Big|
=
	\Big| m^2 \int_\Sigma 
	\phi' \left( N \Delta^k \phi -\Delta^k (N \phi) \right)
	\sqrt{g} \Big|
\\
& \lesssim
	\int_\Sigma \sum_{|I|+1+|J| = k} | m^2 \nabla^{I+1} N \nabla^J \phi| | \sum_{|I'| \leq k} \nabla^{I'} \phi'| \sqrt{g}
	+ \int_M \sum_{|I|+1+|J| = k} |  \nabla^{I+1} N \nabla^J \phi'| \sum_{|I'| \leq k} | m^2 \nabla^{I'} \phi| \sqrt{g}
\\
& \lesssim
	\| m \phi \|_{L^\infty} |\tau|^{-1} \| \tau \phi' \|_{H^k} \| \hN \|_{H^k} 
	+ |\tau|^{-1} \| \tau \phi' \|_{H^k} 
	\sum_{|I|+1+|J|=k} \| \nabla^{I+1}N \|_{H^1}  \| m \nabla^J \phi \|_{H^1}
\\
& \quad
	+|\tau|^{-1} \| \tau \phi' \|_{L^\infty} \|  m\phi \|_{H^k} \| \hN \|_{H^k} 
	+ |\tau|^{-1} \| m \phi \|_{H^k} 
	\sum_{|I|+1+|J|=k} \| \nabla^{I+1}N \|_{H^1} \| \tau \nabla^J \phi' \|_{H^1}
\\
& \lesssim
	|\tau|^{-1} \| \rho \|_{L^\infty}^{1/2}  \Ephi{k}^{1/2} \| \hN \|_{H^k} 
	+ |\tau|^{-1} \tEphi{k}^{1/2} \Big(
		\| \nabla N \|_{H^1} \| \phi \|_{H^k}
		+ \| \phi \|_{H^{k-1}} \| \hN \|_{H^{k+1}} \Big)
\\
& \quad
	+|\tau|^{-1} \varepsilon \tEphi{k}^{1/2} \| \hN \|_{H^k} 
	+ |\tau|^{-1} \tEphi{k}^{1/2}
	\Big( \| \nabla N \|_{H^1} \| \tau \phi'\|_{H^k} + \| \hN \|_{H^{k+1}} \| \tau \phi' \|_{H^{k-1}} \Big)
\\
& \lesssim
	|\tau|^{-1} \varepsilon \tEphi{k}^{1/2} \| \hN \|_{H^k} 
	+ |\tau|^{-1} \tEphi{k}
		\|\hN\|_{H^2} 
		+ |\tau|^{-1} \tEphi{k}^{1/2} \tEphi{k-1}^{1/2} \| \hN \|_{H^{k+1}}.
\end{align*}
Summing from $k=1$ to $\ell$ gives the required result. 
\end{proof}

\begin{rem}
We now outline the key ideas behind closing the Lapse and Klein-Gordon bootstrap assumptions, as proved below in Proposition \ref{prop:KG-lapse-heirarchy}. The estimates for $\tEphi{0}$ and $\| \hN \|_{H^2}$ are readily improved. Then, starting from $\ell=1$, the most problematic terms needed for the $\tEphi{\ell}$ estimate are contained in Lemma \ref{lem:BL-terms}. Nonetheless, Lemma \ref{lem:BL-terms} tells us that we need information about $\| \hN \|_{H^2}$, $\| \hN \|_{H^{\ell+1}}$ and $\tEphi{\ell-1}$, all of  which have been upgraded from the previous steps. The upgraded estimate for $\tEphi{\ell}$ will then be used, via Lemma \ref{lem:lapse-est}, to close the bootstrap estimate for $\| \hN \|_{H^{\ell+2}}$. One then moves onto improving the estimate for $\tEphi{\ell+1}$ and continues until $\ell=N$. 
\end{rem}

\begin{prop}[Upgraded Lapse and Klein-Gordon estimates] \label{prop:KG-lapse-heirarchy}
Assume the bootstrap assumptions \eqref{eq:bootstraps}  hold, then
\begin{align}
\| \hN \|_{H^2} &\lesssim \varepsilon^2 e^{-T} ,
\\
\tEphi{0} & \lesssim \varepsilon^2 ,
\end{align}
and for higher orders $1 \leq  \ell \leq N$
\begin{align}
\| \hN \|_{H^{k+1}} &\lesssim \varepsilon^2 e^{(-1+C \varepsilon)T} ,
\\
\tEphi{k} & \lesssim \varepsilon^2 e^{C \varepsilon T} .
\end{align}

\end{prop}

\begin{proof}
From Proposition \ref{prop:Ezero-KG}
\eq{
\tilde E_0(\phi) \big|_T \lesssim  \tilde E_0(\phi) \big|_{T_0} .
}
The lapse estimate Lemma \ref{lem:lapse-est} with $\ell=2$ then implies 
\eq{
\| \hN \|_{H^2} 
 \lesssim 
	\varepsilon^2 e^{-T}
	+e^{-T} \tEphi{0}
\lesssim \varepsilon^2 e^{-T}.
}
The Klein-Gordon estimate of Proposition \ref{prop:prelim-est-Ek-KG} combined with the borderline estimate of Lemma \ref{lem:BL-terms} for $\ell=2-1=1$ together imply
\begin{align*}
\p_T \tEphi{1}
& \lesssim
	\varepsilon e^{(-1+\gamma)T} 
	\tEphi{1}
	+ |\tau| \tEphi{1}
	+   \varepsilon^3 e^{(-1+\gamma)T}
	+ \varepsilon^4 |\tau|
	+ |\tau|^{-1} \varepsilon \tEphi{1}^{1/2} \| \hN \|_{H^2} 
\\
& \quad
	+ |\tau|^{-1} \tEphi{1}
		\|\hN\|_{H^2} 
		+ |\tau|^{-1} \tEphi{1}^{1/2} \tEphi{0}^{1/2} \| \hN \|_{H^{2}}
\\
& \lesssim
	\varepsilon^3 e^{(-1+\gamma)T}
	+ \varepsilon^4
	+ \big( 
		\varepsilon e^{(-1+\gamma)T}
		+ |\tau| 
		+ \varepsilon^2
		\big) \tEphi{1}.
\end{align*}
Thus Gr\"onwall's inequality implies
\begin{align}
\tEphi{1}\big|_T
&\leq 
	\Big( \tEphi{1}\big|_{T_0} 
	+ C \int_{T_0}^T \big( \varepsilon^3 e^{(-1+\gamma)s}
	+ \varepsilon^4 \big) ds \Big)
	\exp \Big(  C \int_{T_0}^T 
	 \big(e^{-s} + \varepsilon e^{(-1+\gamma) s} + \varepsilon^2\big) ds \Big)
\notag
\\
&\lesssim 
	\big( \tEphi{1}\big|_{T_0} 
	+  \varepsilon^3 e^{C \varepsilon T} \big) \exp(C \varepsilon T).
%\\
%& <
%	\varepsilon^2 e^{C \varepsilon T}.
\end{align}
Note we used the identity: $x \leq 1+x \leq e^x$ for $x \geq 0$. 

Returning to the lapse estimate from Lemma \ref{lem:lapse-est} with $\ell=3$ now implies 
\eq{
\| \hN \|_{H^3} 
 \lesssim 
	\varepsilon^2 e^{-T}
	+e^{-T} \tEphi{1}
\lesssim \varepsilon^2 e^{(-1+2 C \varepsilon )T}.
}
Now we use this result to improve the $\ell=2$ estimate for the Klein-Gordon field. From Lemma \ref{lem:BL-terms}, Proposition \ref{prop:prelim-est-Ek-KG} and the upgraded estimates obtained so far, we have
\begin{align*}
\p_T \tEphi{2}
&\lesssim
	 \varepsilon e^{(-1+\gamma)T} 
	\tEphi{2}
	+ |\tau| \tEphi{2}
	+  \varepsilon^3 e^{(-1+\gamma)T}
	+ \varepsilon^4 |\tau|
	+|\tau|^{-1} \varepsilon \tEphi{2}^{1/2} \| \hN \|_{H^2} 
\\
& \quad
	+ |\tau|^{-1} \tEphi{2}
		\|\hN\|_{H^2} 
		+ |\tau|^{-1} \tEphi{2}^{1/2} \tEphi{1}^{1/2} \| \hN \|_{H^{3}}
\\
& \lesssim
	\varepsilon e^{(-1+\gamma)T} 
	\tEphi{2}
	+ |\tau| \tEphi{2}
	+ \varepsilon^3 e^{(-1+\gamma)T}
	+\varepsilon^3  \tEphi{2}^{1/2} 
	+\varepsilon^2\tEphi{2}
\\
& \quad
		+ \varepsilon^3 e^{C\varepsilon T} \tEphi{2}^{1/2}
\\
& \lesssim
	\varepsilon^3 e^{(-1+\gamma)T}
	+ \varepsilon^4 e^{2C\varepsilon T} 
	+ \big(\varepsilon e^{(-1+\gamma)T}
	+ |\tau|
	+ \varepsilon^2 \big)\tEphi{2}.
\end{align*}
Thus by Gr\"onwall:
\begin{align*}
\tEphi{2}\big|_T
&\lesssim
	\Big( \tEphi{2}\big|_{T_0} 
	+ \int_{T_0}^T \big(\varepsilon^3  e^{(-1+\gamma)s}   + \varepsilon^4 e^{C \varepsilon s} ds\big)  \Big)
	\exp \Big( C \int_{T_0}^T (e^{-s} + \varepsilon)  ds \Big)
\\
&\lesssim 
	\big( \tEphi{2}\big|_{T_0} 
	+ \varepsilon^3 e^{C \varepsilon T} \big) e^{C \varepsilon T}.
%\\
%& <
%	\varepsilon^2 e^{C \varepsilon T}.
\end{align*}
Continuing in this way we stop at $\ell=N+1$ for the lapse estimate. 
\eq{
\| \hN \|_{H^{N+1}} 
 \lesssim 
	\varepsilon^2 e^{-T}
	+e^{-T} \tEphi{N-1}
\lesssim
	 \varepsilon^2 e^{(-1+C\varepsilon )T}.
}
Using this to estimate the final $\ell = N$ energy (recall $N \geq 4$) for the Klein-Gordon field gives
\begin{align*}
\p_T \tEphi{N}
&\lesssim
	 \varepsilon e^{(-1+C \varepsilon)T} 
	\tEphi{N}
	+ |\tau| \tEphi{N}
	+ \varepsilon^3 e^{(-1+C \varepsilon)T}
	+ |\tau| \varepsilon^4
	+ |\tau|^{-1} \varepsilon \tEphi{N}^{1/2} \| \hN \|_{H^N} 
\\
& \quad
	+ |\tau|^{-1} \tEphi{N}
		\|\hN\|_{H^2} 
		+ |\tau|^{-1} \tEphi{N}^{1/2} \tEphi{N-1}^{1/2} \| \hN \|_{H^{N+1}}
\\
&\lesssim
	\varepsilon^3 e^{(-1+C \varepsilon)T}
	+|\tau|
	\tEphi{N}
	+ \varepsilon^3 e^{C \varepsilon T}  \tEphi{N}^{1/2}
	+ \varepsilon^2 \tEphi{N}
\\
&\lesssim
	\varepsilon^3 e^{(-1+C \varepsilon)T}
	+ \varepsilon^4 e^{2C \varepsilon T} 
	+\big( \varepsilon e^{(-1+C \varepsilon)T} 
	+ |\tau| + \varepsilon^2 \big) \tEphi{N}.
\end{align*}
Thus by Gr\"onwall we have
\begin{align*}
\tEphi{N}\big|_T
&\lesssim
	\Big( \tEphi{N}\big|_{T_0} 
	+ \int_{T_0}^T \big(\varepsilon^2  e^{(-1+C \varepsilon)s}   + \varepsilon^4 e^{C \varepsilon s} \big)ds  \Big)
	\exp \Big( C \int_{T_0}^T (\varepsilon + e^{-s} ) ds \Big)
\\
&\lesssim 
	\big( \tEphi{N}\big|_{T_0} 
	+ \varepsilon^3 e^{C \varepsilon T} \big) e^{C \varepsilon T}.
%\\
%& <
%	\varepsilon^2 e^{C \varepsilon T}.
\end{align*}
\end{proof}

\section{Energy estimate - Geometry}
In this final section we obtain improved estimates for the shift vector field and the second fundamental form and metric perturbation. 
\begin{corol}[Improved Shift estimate]
Assume the bootstrap assumptions \eqref{eq:bootstraps} hold, then for $3 \leq \ell \leq N+1$
 \eq{
\| X \|_{H^\ell} \lesssim \varepsilon^2 e^{(-1+C\varepsilon) T}.
}
\end{corol}
\begin{proof}
This follows clearly from Lemma \ref{lem:shift-est} and Proposition \ref{prop:KG-lapse-heirarchy}.
\end{proof}

Recall the definitions from Section \ref{subsec:energy-geom-def}. Using the energy estimate given in \cite[Lemma 20]{AF17}, itself adapted from \cite{AnMo11} we have the following estimate for the second fundamental form and metric perturbation.  

\begin{corol}[Improved geometry estimate]
Assume the bootstrap assumptions \eqref{eq:bootstraps} hold, then for $1 \leq \ell \leq N+1$
\eq{ \label{eq:geom-est-AF}
\p_T \Eg_\ell 
\leq	
	- 2 \alpha \Eg_\ell + 6 \Eg_\ell^{1/2} |\tau| \| N S \|_{H^{\ell-1}} + C \Eg_\ell^{3/2}
	+ C \Eg_\ell^{1/2} \big( 
	|\tau| \| \eta \|_{H^{\ell-1}} 
	+ \tau^2 \| N j \|_{H^{\ell-2}}
	\big).
}
Furthermore
\begin{align*}
\p_T \Eg_\ell 
&\leq
	- 2 \alpha \Eg_\ell 
	+ C \Eg_\ell^{1/2} \varepsilon^2 e^{(-1+C\varepsilon)T}
	 + C \Eg_\ell^{3/2} ,
\end{align*}
Finally 
\eq{\label{final-decay-en}
\Eg_\ell \big|_T
\leq C \varepsilon^2 e^{-2\alpha \zeta T} .
}
where for sufficiently small $\varepsilon$ we may choose $\zeta$ arbitrarily close to 1, in particular  $\zeta \leq 1- \tfrac{C \varepsilon}{\alpha}$.

\end{corol}
\begin{proof}
The estimate \eqref{eq:geom-est-AF} comes from \cite[Lemma 20]{AF17}. Recall $S_{ij}$ from \eqref{eq:Sij}. 
So for $2 \leq \ell \leq N+1$ we find
\begin{align*}
&\| N S \|_{H^{\ell-1}}
\\
& \lesssim
	\| N \|_{L^\infty}
	\Big( \| \phi \|_{L^\infty} \| \phi \|_{H^{\ell-1}} + \| \phi \|_{H^{\ell-1}}^2
	+ \| \tau \nabla \phi \|_{L^\infty} \| \tau \nabla \phi \|_{H^{\ell-1}}
	+ \| \tau \nabla \phi \|_{H^{\ell-1}}^2 
	\Big)
\\
& 
	+ \| \hN\|_{H^{\ell-1}}
	\Big(\| \phi \|_{L^\infty}^2 + \| \phi \|_{L^\infty}^2 \| \phi \|_{H^{\ell-1}} + \| \phi \|_{H^{\ell-1}}^2
	+ \| \tau \nabla \phi \|_{L^\infty}^2
	+  \| \tau \nabla \phi \|_{L^\infty} \| \tau \nabla \phi \|_{H^{\ell-1}}
	+ \| \tau \nabla \phi \|_{H^{\ell-1}}^2 
	\Big)
\\
& \lesssim
	\| \rho \|_{L^\infty}
	+ \tEphi{\ell-1}
	+ \| \hN \|_{H^{\ell-1}} 
	\big( \| \rho \|_{L^\infty}
	+ \tEphi{\ell-1}
	\big)
\end{align*}
where we took note of the product identity \eqref{eq:Hk-square-product}. Also following the method of Lemmas \ref{lem:lapse-est} and \ref{lem:shift-est} we have
 \begin{align*}
|\tau| \| \eta \|_{H^{\ell-1}} 
	+ \tau^2 \| N j \|_{H^{\ell-2}}
& \lesssim
	|\tau| \tEphi{\ell-1} + |\tau| \| \rho \|_{L^\infty}
	+ |\tau| \big(  \| \rho \|_{L^\infty}
	+  \tEphi{\ell-2}\big)
	\big( \| N \|_{L^\infty} + \| \hN \|_{H^\ell} \big).
\end{align*}
Putting these together gives
\begin{align*}
\p_T \Eg_\ell 
& \leq	
	- 2 \alpha \Eg_\ell + 6 \Eg_\ell^{1/2} \tau \| N S \|_{H^{\ell-1}} + C \Eg_\ell^{3/2}
	+ C \Eg_\ell^{1/2} \big( 
	|\tau| \| \eta \|_{H^{\ell-1}} 
	+ \tau^2 \| N j \|_{H^{\ell-2}}
	\big)
\\
& \leq
	- 2 \alpha \Eg_\ell 
	+ 6 \Eg_\ell^{1/2} |\tau| 
	\big( \| \rho \|_{L^\infty} + \tEphi{\ell-1} \big) 
	 + C \Eg_\ell^{3/2}
\\
& \quad
	+ C \Eg_\ell^{1/2} \Big(
	|\tau| \tEphi{\ell-1} + |\tau| \| \rho \|_{L^\infty}
	+ |\tau| \| \rho \|_{L^\infty} \| \hN \|_{H^\ell}
	+ |\tau| \tEphi{\ell-2}\| \hN \|_{H^\ell} \Big).
\end{align*}
Thus
\begin{align*}
\p_T \big( \Eg_\ell ^{1/2} \big)
&\leq
	- \alpha \Eg_\ell ^{1/2}
	+ C \varepsilon^2 e^{-T}
	 + C  \Eg_\ell.
\end{align*}
Now recall $\alpha \in [1-\delta_\alpha, 1]$ where $\delta_\alpha$ can be made suitably small. Given $\alpha$, pick $\zeta$ such that $\tfrac34 < \zeta <1$ and $-(\alpha \zeta - \tfrac34) < 0$ (ie, $\alpha \zeta > \tfrac34$). Indeed we can guarantee $\alpha \zeta > \tfrac34$ holds by choosing $\varepsilon$ sufficiently small such that $1-\delta_\alpha(\varepsilon) > {3 \over 4 \zeta}$. Then we have
\begin{align*}
\p_T \big( e^{\tfrac34 T} \Eg_\ell ^{1/2} \big)
&\leq
	- (\alpha \zeta - \tfrac34) e^{\tfrac34 T}\Eg_\ell ^{1/2}
	+ C \varepsilon^2 e^{(-1+\tfrac34)T}
	- e^{\tfrac34 T} \Eg_\ell^{1/2}
	\Big( \alpha (1-\zeta) - C  \Eg_\ell^{1/2}\Big)
\\
&\leq
	- (\alpha \zeta - \tfrac34) e^{\tfrac34 T}\Eg_\ell ^{1/2}
	+ C \varepsilon^2 e^{(-1+\tfrac34)T}
	- e^{\tfrac34 T} \Eg_\ell^{1/2}
	\Big( \alpha (1-\zeta) - C  \varepsilon \Big)
\\
&\leq
	- (\alpha \zeta - \tfrac34) e^{\tfrac34 T}\Eg_\ell ^{1/2}
	+ C \varepsilon^2 e^{(-1+\tfrac34)T}
\end{align*}
where we dropped the final term by picking $\varepsilon$ small enough so that $\alpha(1-\zeta)- C \varepsilon \geq 0$. 
Then by Gr\"onwall we have
\eq{
e^{\delta T} \Eg_\ell ^{1/2} \big|_T
\leq \Big( e^{\tfrac34 T_0} \Eg_\ell ^{1/2} \big|_{T_0} + C \varepsilon^2 \int_{T_0}^T e^{(-1+\tfrac34)s} ds \Big)\Big)
\exp \Big( - \int_{T_0}^T  (\alpha \zeta -\tfrac34) ds \Big).
}
This implies 
\eq{
\Eg_\ell ^{1/2} \big|_T
\leq 
	\big( e^{\tfrac34 T_0} \Eg_\ell ^{1/2} \big|_{T_0} + C \varepsilon^2 \big)
	 e^{-\tfrac34 T} e^{-(\alpha \zeta - \tfrac34)(T-T_0)}
}
and thus
\eq{
\Eg_\ell \big|_T
\leq (\Eg_\ell |_{T_0}  + C \varepsilon^4 )e^{-2\alpha \zeta T}
< C \varepsilon^2 e^{-\tfrac32 T}.
}

\end{proof}

\begin{proof}[Proof of the main theorem]
The main theorem is now a consequence of the foregoing lemmas. Considering a sufficiently small perturbation of the Milne initial data it follows along the lines of the corresponding argument in \cite{FK15} that there exists a CMC surface in the maximal development with initial data close to the Milne geometry. This initial data set is now evolved by the rescaled CMCSH EKGS. The local existence theory then implies the existence of a solution, which according to our previous analysis obeys the decay estimates for the perturbation given in \eqref{final-decay-en}.
In particular, the solution exists for $T\rightarrow\infty$ and moreover is future complete, which follows analogous to \cite{AnMo11}.
\end{proof}


\begin{thebibliography}{10}
\providecommand{\url}[1]{\texttt{#1}}
\expandafter\ifx\csname urlstyle\endcsname\relax
  \providecommand{\doi}[1]{doi: #1}\else
  \providecommand{\doi}{doi: \begingroup \urlstyle{rm}\Url}\fi




\bibitem[AF]{AF17}
\textsc{Andersson}, L.~; \textsc{Fajman}, D.~:
\newblock {Nonlinear stability of the Milne model with matter}
\newblock arXiv:1709.00267, 2017

\bibitem[AMa]{AnMo03}
\textsc{Andersson}, L.~; \textsc{Moncrief}, V.~:
\newblock {Elliptic-hyperbolic systems and the {E}instein equations}
\newblock \emph{Ann.~Henri Poincar\'e}, \textbf 4 (2003)

\bibitem[AMb]{AnMo04}
\textsc{Andersson}, L.~; \textsc{Moncrief}, V.~:
\newblock {Future complete vacuum spacetimes}
\newblock \emph{The Einstein equations and the Large Scale Behavior of Gravitational Fields -- 50 years of the Cauchy problem in General Relativity},  Editors: Piotr T. Chru\'sciel, Helmut Friedrich (2004), Birkh\"auser

\bibitem[AMc]{AnMo11}
\textsc{Andersson}, L.~; \textsc{Moncrief}, V.~:
\newblock {Einstein spaces as attractors for the Einstein flow.}
\newblock {In: }\emph{J. Differ. Geom.} \textbf{89} (2011), no. 1, 1--47


\bibitem[B]{Be08}
\textsc{Besse}, A.~L.:
\newblock \emph{{Einstein manifolds. Reprint of the 1987 edition.}}
\newblock {Berlin: Springer}, (2008)

\bibitem[BZ]{BZ}
\textsc{Bieri}, L.; \textsc{Zipser}, N.
\newblock \emph{Extensions of the stability theorem of the Minkowski space in general relativity}
\newblock American Mathematical Society, International Press, (2009)

\bibitem[BFK]{BFK}
\textsc{Branding}, V. ; \textsc{Fajman}, D. ; \textsc{Kr\"oncke}, K.
\newblock \emph{Stable cosmological {K}aluza-{K}lein Spacetimes}
\newblock to appear in Commun.~Math.~Phys.~, (2019)

\bibitem[CK]{CK93}
\textsc{Christodoulou}, D. ; \textsc{Klainerman}, S.:
\newblock \emph{The global nonlinear stability of the Minskowski space}
\newblock Princeton University Press, (1993)

\bibitem[CM]{ChMo01}
\textsc{Choquet-Bruhat}, Y. ; \textsc{Cotsakis}, S.:
\newblock \emph{Global hyperbolicity and completeness}
\newblock J.~Geom.~Phys. \textbf{43}, (2002)

\bibitem[CC]{CC02}
\textsc{Choquet-Bruhat}, Y. ; \textsc{Moncrief}, V.:
\newblock \emph{Future Global in Time Einsteinian Spacetimes with $U(1)$ Isometry Group}
\newblock Ann.~Henri Poincar\'e, \textbf{2}, (2001)

\bibitem[CCL]{CCL}
\textsc{Choquet-Bruhat}, Y. ; Chru\'sciel, P. ; Loizelet, J.
\newblock \emph{Global solutions of the {E}instein-{M}axwell equations in
              higher dimensions}
\newblock Class.~Quant.~Grav., \textbf{23}  (2006)

%
\bibitem[CB]{CB09}
\textsc{Choquet-Bruhat}, Y. :
\newblock General Relativity and the Einstein equations
\newblock Oxford Mathematical Monographs, Oxford University Press 

\bibitem[Fa]{Faa}
\textsc{Fajman}, D.:
\newblock The nonvacuum Einstein flow on surfaces of negative curvature and nonlinear stability,
\newblock \emph{Comm.~Math.~Phys.~}\textbf{353}, 2, (2017)

\bibitem[Fb]{Fa16}
\textsc{Fajman}, D.:
\newblock Locall well-posedness for the {E}instein-{V}lasov  system,
\newblock SIAM J.~Math.~Anal.~\textbf{48}, 2016 

\bibitem[FJS]{FJS17-2}
\textsc{Fajman}, D. ; \textsc{Joudioux}, J. ; \textsc{Smulevici}, J.,
\newblock{The stability of the Minkowki space for the Einstein-Vlasov system},
\newblock arXiv:1707.06141 , (2017)

\bibitem[FK]{FK15}
\textsc{Fajman}, D. ; \textsc{Kr\"oncke}, K.
\newblock{Stable fixed points of the Einstein flow with a positive cosmological constant}, arXiv:1504.00687, 2015

\bibitem[Fr]{F}
\textsc{Friedrich}, H.
\newblock {On the existence of {$n$}-geodesically complete or future complete solutions of {E}instein's field equations with smooth asymptotic structure}
\newblock Comm.~Math.~Phys. \textbf{107} (1986)

\bibitem[HV]{HiVa18}
\textsc{Hintz}, P. ; \textsc{Vasy}, A. 
\newblock{The global non-linear stability of the {K}err--de {S}itter
              family of black holes}
\newblock Acta Math.~ \textbf{220} (2018) 

\bibitem[HS]{HS13}
\textsc{Holzegel}, G.~; \textsc{Smulevici}, J.~:
\newblock Decay properties of Klein-Gordon fields on Kerr-AdS spacetimes
\newblock \emph{Commun.~Pure~Appl.~Math.}
\textbf{66}, (2013), 1751-1802 


\bibitem[LM]{LM}
\textsc{LeFloch}, P. ; \textsc{Ma}, Y.
\newblock \emph{The global nonlinear stability of Minkowski space for self-gravitating massive fields}
\newblock WSP. \textbf{3} (2018)


\bibitem[LR]{LR}
\textsc{Lindblad}, H. ; \textsc{Rodnianski}, I.
\newblock \emph{The global stability of Minkowski space-time in harmonic gauge}
\newblock Ann.~of Math. (2), 171, (2010)

\bibitem[LT]{LT}
\textsc{Lindblad}, H. ; \textsc{Taylor}, M.
\newblock \emph{Global stability of Minkowski space for the Einstein--Vlasov system in the harmonic gauge}
\newblock arXiv:1707.06079 (2017)


\bibitem[RS]{RoSp18}
\textsc{Rodnianski}, I.; \textsc{Speck}, J.:
\newblock{A regime of linear stability for the {E}instein-scalar field
              system with applications to nonlinear big bang formation}
\newblock \emph{Ann. of Math. (2)} \textbf 187, (2018)



\bibitem[Re]{Re08}
\textsc{Rendall}, A.~D.:
\newblock \emph{{Partial Differential Equations in General Relativity}}
\newblock {Oxford Graduate Texts in Mathematics}, (2008)



\bibitem[Ri]{Ri08}
\textsc{Ringstr\"om}, H.~:
\newblock \emph{{Future stability of the {E}instein-non-linear scalar field
              system}}
\newblock {Invent.~Math.~}, \textbf{173} (2008)


\bibitem[S]{S}
\textsc{Speck}, J.:
\newblock \emph{The global stability of the Minkowski spacetime solution to the Einstein-nonlinear system in wave coordinates}
\newblock Anal.~PDE, \textbf 7, 2014

\bibitem[T]{Ta16}
\textsc{Taylor}, M.~:
\newblock \emph{The global nonlinear stability of Minkowski space for the massless Einstein-Vlasov system} 
\newblock Annals of PDE \textbf{3}, 9 2017

\bibitem[W-J]{Wj18}
\textsc{Wang}, J.:
\newblock {Future stability of the $1+3$ Milne model for Einstein-Klein-Gordon system}
\newblock arXiv:1805.01106, (2018)

\bibitem[W-Q]{Wq18}
\textsc{Wang}, Q.:
\newblock An intrinsic hyperboloid approach for Einstein Klein-Gordon equations
\newblock \emph{Journal of Differential Geometry.}

\bibitem[Wy]{Wy}
\textsc{Wyatt}, Z.:
\newblock The weak-null condition and Kaluza-Klein Spacetimes
\newblock \emph{J.~Hyp.~Diff.~Eq.~} \textbf 15, 2018


\end{thebibliography}
\end{document}